\documentclass[envcountsame,UKenglish]{llncs}
\usepackage{ifthen}
\newcommand{\techRep}{true} 
\newcommand{\iftechrep}{\ifthenelse{\equal{\techRep}{true}}}

\usepackage{amsmath,amssymb}
\usepackage{tikz}
\usetikzlibrary{arrows,calc,automata}

\usepackage{subfig}

\newcommand{\conv}{\mathit{conv}}

\newcommand{\spco}[1]{#1^*}

\newcommand{\R}{\mathbb{R}}
\newcommand{\N}{\mathbb{N}}

\newcommand{\norm}[1]{\lVert#1\rVert}
\newcommand{\A}{\mathcal{A}}
\newcommand{\AF}{\overrightarrow{\mathcal{A}}}
\newcommand{\AB}{\overleftarrow{\mathcal{A}}}
\newcommand{\AFB}{\overleftarrow{\overrightarrow{\mathcal{A}}}}
\newcommand{\ABF}{\overrightarrow{\overleftarrow{\mathcal{A}}}}
\newcommand{\alphaF}{\overrightarrow{\alpha}}

\newcommand{\B}{\mathcal{B}}
\newcommand{\BV}{\mathsf{B}}
\newcommand{\Bh}{\widehat{B}}

\newcommand{\E}{\mathcal{E}}
\newcommand{\mache}{\varepsilon_\text{mach}}
\newcommand{\etaB}{\overleftarrow{\eta}}
\newcommand{\etaF}{\overrightarrow{\eta}}

\newcommand{\ExThR}{\mathit{ExTh}(\R)}
\newcommand{\FV}{\mathsf{F}}

\newcommand{\Ft}{\widetilde{F}}
\newcommand{\Fh}{\widehat{F}}
\newcommand{\Fs}[1]{F_{#1}}
\newcommand{\nF}{\overrightarrow{n}}
\newcommand{\nB}{\overleftarrow{n}}
\newcommand{\nBF}{\overrightarrow{\overleftarrow{n}}}

\newcommand{\lambdas}[1]{\lambda^{(#1)}}

\newcommand{\MF}{\overrightarrow{M}}

\newcommand{\MB}{\overleftarrow{M}}

\newcommand{\Oti}{\ensuremath{\widetilde{O}}}
\newcommand{\Ps}[1]{P_{#1}}

\newcommand{\rank}{\ensuremath{\textup{rank}}}

\newcommand{\sign}{\mathit{sign}}
\newcommand{\supp}{\mathit{supp}}
\newcommand{\ind}{\mbox{}\hspace{7mm}}
\newcommand{\indd}{\ind\ind}
\newcommand{\inddd}{\indd\ind}
\newcommand{\indddd}{\inddd\ind}
\newcommand{\inddddd}{\indddd\ind}
\newcommand{\set}[1]{\N_{#1}}

\newcommand{\hatF}{\widehat{F}}
\newcommand{\hatf}{\widehat{f}}
\newcommand{\hatg}{\widehat{g}}

\newenvironment{qtheorem}[1]{%
{\par\medskip\noindent\bf Theorem #1.}
\begin{itshape}%
}{%
\end{itshape}%
}

\newenvironment{qproposition}[1]{%
{\par\medskip\noindent\bf Proposition #1.}
\begin{itshape}%
}{%
\end{itshape}%
}

\iftechrep{
\title{Stability and Complexity of Minimising Probabilistic Automata}
}{
\title{Stability and Complexity of Minimising Probabilistic Automata\thanks{For a full version of this paper, see~\cite{KW14-icalp-report}.}}
}
\titlerunning{Stability and Complexity of Minimising Probabilistic Automata}

\author{Stefan Kiefer \and Bj\"orn Wachter}
\institute{University of Oxford, UK}

\begin{document}
\maketitle
\pagestyle{plain}

\begin{abstract}
We consider the state-minimisation problem for weighted and probabilistic automata.
We provide a numerically stable polynomial-time minimisation algorithm for weighted automata,
with guaranteed bounds on the numerical error when run with floating-point arithmetic.
Our algorithm can also be used for ``lossy'' minimisation with bounded error.
We show an application in image compression.
In the second part of the paper we study the complexity of the minimisation problem for probabilistic automata.
We prove that the problem is NP-hard and in PSPACE, improving a recent EXPTIME-result.
\end{abstract}

\section{Introduction} \label{sec-intro}

Probabilistic and weighted automata were introduced in the 1960s, with many fundamental results established by
Sch\"{u}tzenberger~\cite{Schutzenberger} and Rabin~\cite{Rab63}.
Nowadays probabilistic automata are widely used in automated verification, natural-language processing, and machine learning.

Probabilistic automata (PAs) generalise deterministic finite automata (DFAs):
The transition relation specifies, for each state~$q$ and each input letter~$a$, a probability distribution on the successor state.
Instead of a single initial state, a PA has a probability distribution over states;
 and instead of accepting states, a PA has an acceptance probability for each state.
As a consequence, the language induced by a PA is a \emph{probabilistic language},
 i.e., a mapping $L : \Sigma^* \to [0,1]$, which assigns each word an acceptance probability.
Weighted automata (WAs), in turn, generalise PAs:
 the numbers appearing in the specification of a WA may be arbitrary real numbers.
As a consequence, a WA induces a \emph{weighted language}, i.e., a mapping $L : \Sigma^* \to \R$.
Loosely speaking, the weight of a word~$w$ is the sum of the weights of all accepting $w$-labelled paths through the WA.

Given an automaton, it is natural to ask for a small automaton that accepts the same weighted language.
A small automaton is particularly desirable when further algorithms are run on the automaton,
 and the runtime of those algorithms depends crucially on the size of the automaton~\cite{KMOWW:CAV11}.
In this paper we consider the problem of minimising the number of states of a given WA or PA, while preserving its (weighted or probabilistic) language.

WAs can be minimised in polynomial time, using, e.g., the standardisation procedure of~\cite{Schutzenberger}.
When implemented efficiently (for instance using triangular matrices), one obtains an $O(|\Sigma| n^3)$ minimisation algorithm, where $n$ is the number of states.
As PAs are special WAs, the same holds in principle for PAs.

There are two problems with these algorithms:
 (1) numerical instability, i.e., round-off errors can lead to an automaton that is not minimal and/or induces a different probabilistic language;
 and (2) minimising a PA using WA minimisation algorithms does not necessarily result in a PA:
  transition weights may, e.g., become negative.
This paper deals with those two issues.

Concerning problem~(1), numerical stability is crucial under two scenarios:
 (a) when the automaton size makes the use of exact rational arithmetic prohibitive, and thus necessitates floating-point arithmetic~\cite{KMOWW:CAV11};
 or (b) when exact minimisation yields an automaton that is still too large and a ``lossy compression'' is called for,
  as in image compression~\cite{CulikKari93}.
Besides finding a numerically stable algorithm, we aim at two further goals:
First, a stable algorithm should also be efficient; i.e., it should be as fast as classical (efficient, but possibly unstable) algorithms.
Second, stability should be provable, and ideally there should be easily computable error bounds.
In Section~\ref{sec-minwei} we provide a numerically stable $O(|\Sigma| n^3)$ algorithm for minimising WAs.
The algorithm generalises the \emph{Arnoldi iteration}~\cite{Arnoldi51} which is used for locating eigenvalues
in numerical linear algebra.
The key ingredient, leading to numerical stability and allowing us to give error bounds, is the use of
special orthonormal matrices, called \emph{Householder reflectors}~\cite{Householder58}.
To the best of the authors' knowledge, these techniques have not been previously utilised for computations on weighted automata.

Problem~(2) suggests a study of the computational complexity of the \emph{PA~minimisation problem}:
 given a PA and $m \in \N$, is there an equivalent PA with $m$ states?
In the 1960s and 70s, PAs were studied extensively, see the survey~\cite{Bukharaev80} for references
 and Paz's influential textbook~\cite{Paz71}.
PAs appear in various flavours and under different names.
%
For instance, in \emph{stochastic sequential machines}~\cite{Paz71} there is no fixed initial state distribution,
 so the semantics of a stochastic sequential machine is not a probabilistic language, but a mapping from initial distributions to probabilistic languages.
This gives rise to several notions of minimality in this model~\cite{Paz71}.
In this paper we consider only PAs with an initial state distribution; equivalence means equality of probabilistic languages.

One may be tempted to think that PA~minimisation is trivially in NP, by guessing the minimal PA and verifying equivalence.
However, it is not clear that the minimal PA has rational transition probabilities, even if this holds for the original PA.

For DFAs, which are special PAs, an automaton is minimal (i.e., has the least number of states)
 if and only if all states are reachable and no two states are equivalent.
However, this equivalence does in general not hold for PAs.
In fact, even if a PA has the property that no state behaves like a convex combination of other states,
 the PA may nevertheless not be minimal. 
As an example,
consider the PA in the middle of Figure~\ref{fig-redCubeAut} on page~\pageref{fig-redCubeAut}.
State~$3$ behaves like a convex combination of states $2$ and~$4$:
 state~$3$ can be removed by splitting its incoming arc with weight~$1$ in two arcs with weight $1/2$ each and redirecting the new arcs to states $2$ and~$4$.
The resulting PA is equivalent and no state can be replaced by a convex combination of other states.
But the PA on the right of the figure is equivalent and has even fewer states.

In Section~\ref{sec-minpro} we show that the PA minimisation problem is NP-hard by a reduction from 3SAT.
A step in our reduction is to show that the following problem, the \emph{hypercube problem}, is NP-hard:
 given a convex polytope~$P$ within the $d$-dimensional unit hypercube and $m \in \N$,
  is there a convex polytope with $m$ vertices that is nested between $P$ and the hypercube?
We then reduce the hypercube problem to PA minimisation.
To the best of the authors' knowledge,
 no lower complexity bound for PA minimisation has been previously obtained,
 and there was no reduction from the hypercube problem to PA minimisation.
However, towards the converse direction,
 the textbook~\cite{Paz71} suggests that an algorithm for the hypercube problem could serve as a
 ``subroutine'' for a PA minimisation algorithm, leaving the decidability of both problems open.
In fact, problems similar to the hypercube problem were subsequently studied in the field of computational geometry,
 citing PA minimisation as a motivation~\cite{Silio79,Mitchell92,DasJoseph92,DasGoodrich95}.

The PA minimisation problem was shown to be decidable in~\cite{MateusQL12},
 where the authors provided an exponential reduction to the existential theory of the reals,
  which, in turn, is decidable in PSPACE~\cite{Can88,Renegar92}, but not known to be PSPACE-hard.
In Section~\ref{sub-PSPACE} we give a polynomial-time reduction from the PA minimisation problem
 to the existential theory of the reals.
It follows that the PA minimisation problem is in PSPACE, improving the EXPTIME result of~\cite{MateusQL12}.

\section{Preliminaries} \label{sec-prelim}

In the technical development that follows it is more convenient to talk about vectors
 and transition matrices than about states, edges, alphabet labels and weights.
However, a PA ``of size~$n$'' can be easily viewed as a PA with states $1, 2, \ldots, n$.
We use this equivalence in pictures.

Let $\N = \{0, 1, 2, \ldots\}$.
For $n \in \N$ we write $\set{n}$ for the set $\{1, 2, \ldots, n\}$.
For $m, n \in \N$, elements of $\R^m$ and~$\R^{m \times n}$ are viewed as vectors and matrices, respectively.
Vectors are row vectors by default.
Let $\alpha \in \R^m$ and $M \in \R^{m \times n}$.
We denote the entries by $\alpha[i]$ and~$M[i,j]$ for $i \in \N_m$ and $j \in \N_n$.
By~$M[i,\cdot]$ we refer to the $i$th row of~$M$.
By~$\alpha[i..j]$ for $i \le j$ we refer to the sub-vector $(\alpha[i], \alpha[i+1], \ldots, \alpha[j])$, and similarly for matrices.
We denote the transpose by $\alpha^T$ (a column vector) and $M^T \in \R^{n \times m}$.
We write $I_n$ for the $n \times n$ identity matrix.
When the dimension is clear from the context, we write $e(i)$ for the vector with $e(i)[i] = 1$ and $e(i)[j] = 0$ for $j \ne i$.
A vector $\alpha \in \R^m$ is \emph{stochastic} if $\alpha[i] \ge 0$ for all $i \in \set{m}$ and $\sum_{i=1}^m \alpha[i] \le 1$.
A matrix is \emph{stochastic} if all its rows are stochastic.
By $\norm{\cdot} = \norm{\cdot}_2$, we mean the 2-norm for vectors and matrices throughout the paper unless specified otherwise.
If a matrix~$M$ is stochastic, then $\norm{M} \le \norm{M}_1 \le 1$.
For a set $V \subseteq \R^n$, we write $\langle V \rangle$ to denote the vector space spanned by~$V$,
 where we often omit the braces when denoting~$V$.
For instance, if $\alpha, \beta \in \R^n$, then $\langle \{\alpha, \beta \} \rangle = \langle \alpha, \beta \rangle = \{r \alpha + s \beta \mid r, s \in \R\}$.

An \emph{$\R$-weighted automaton (WA)} $\A = (n, \Sigma, M, \alpha, \eta)$ consists of a size $n \in \N$,
 a finite alphabet~$\Sigma$,
 a map $M : \Sigma \to \R^{n \times n}$, an initial (row) vector $\alpha \in \R^n$, and a final (column) vector $\eta \in \R^n$.
Extend $M$ to $\Sigma^*$ by setting \hbox{$M(a_1 \cdots a_k) := M(a_1) \cdots M(a_k)$}.
The \emph{language~$L_\A$} of a WA~$\A$ is the mapping $L_\A : \Sigma^* \to \R$ with $L_\A(w) = \alpha M(w) \eta$.
WAs $\A, \B$ over the same alphabet~$\Sigma$ are said to be \emph{equivalent} if $L_\A = L_\B$.
A WA~$\A$ is \emph{minimal} if there is no equivalent WA~$\B$ of smaller size.

A \emph{probabilistic automaton (PA)} $\A = (n, \Sigma, M, \alpha, \eta)$ is a WA,
 where $\alpha$ is stochastic,
 $M(a)$ is stochastic for all $a \in \Sigma$,
 and $\eta \in [0,1]^n$.
A PA is a \emph{DFA} if all numbers in $M, \alpha, \eta$ are $0$ or~$1$.

\section{Stable WA Minimisation} \label{sec-minwei}

In this section we discuss  WA minimisation.
In Section~\ref{sub-brzozowski} we describe a WA minimisation algorithm in terms of elementary linear algebra.
The presentation reminds of Brzozowski's algorithm for NFA minimisation~\cite{Brzozowski62}.%
\footnote{In~\cite{BonchiPrakash} a very general Brzozowski-like minimization algorithm is presented in terms of universal algebra.
One can show that it specialises to ours in the WA setting.}
WA minimisation techniques are well known, originating in~\cite{Schutzenberger}, cf.\ also \cite[Chapter II]{BerstelReutenauer} and~\cite{Beimel00}.
Our algorithm and its correctness proof may be of independent interest, as they appear to be particularly succinct.
In Sections \ref{sub-arnoldi}~and~\ref{sub-min-image} we take further advantage of the linear algebra setting
 and develop a numerically stable WA minimisation algorithm.

\subsection{Brzozowski-like WA Minimisation} \label{sub-brzozowski}

Let $\A = (n, \Sigma, M, \alpha, \eta)$ be a WA.
Define the \emph{forward space} of~$\A$ as the (row) vector space
 $\FV := \langle \alpha M(w) \mid w \in \Sigma^* \rangle$.
Similarly, let the \emph{backward space} of~$\A$ be the (column) vector space
 $\BV := \langle M(w) \eta \mid w \in \Sigma^* \rangle$.
Let $\nF \in \N$ and $F \in \R^{\nF \times n}$ such that the rows of~$F$ form a basis of~$\FV$.
Similarly, let $\nB \in \N$ and $B \in \R^{n \times \nB}$ such that the columns of~$B$ form a basis of~$\BV$.
Since $\FV M(a) \subseteq \FV$ and $M(a) \BV \subseteq \BV$ for all $a \in \Sigma$,
 there exist maps $\MF : \Sigma \to \R^{\nF \times \nF}$ and $\MB : \Sigma \to \R^{\nB \times \nB}$ such that
 \begin{equation}
  F M(a) = \MF(a) F \quad \text{and} \quad M(a) B = B \MB(a) \quad \text{for all $a \in \Sigma$.}
   \label{eq-commutativity}
 \end{equation}
We call $(F, \MF)$ a \emph{forward reduction} and $(B, \MB)$ a \emph{backward reduction}.
We will show that minimisation reduces to computing such reductions.
By symmetry we can focus on forward reductions.
We call a forward reduction $(F, \MF)$ \emph{canonical} if
 $F[1,\cdot]$ (i.e., the first row of~$F$) is a multiple of~$\alpha$,
 and the rows of~$F$ are orthonormal, i.e., $F F^T = I_{\nF}$.

Let $\A = (n, \Sigma, M, \alpha, \eta)$ be a WA with forward and backward reductions $(F,\MF)$ and $(B,\MB)$, respectively.
Let $\alphaF \in \R^{\nF}$ be a row vector such that \hbox{$\alpha = \alphaF F$};
let $\etaB \in \R^{\nB}$ be a column vector such that $\eta = B \etaB$.
(If $(F,\MF)$ is canonical, we have $\alphaF = (\pm \norm{\alpha},0,\ldots,0)$.)
Call $\AF := (\nF, \Sigma, \MF, \alphaF, F \eta)$ a \emph{forward WA} of~$\A$ with base~$F$
 and $\AB := (\nB, \Sigma, \MB, \alpha B, \etaB)$ a \emph{backward WA} of~$\A$ with base~$B$.
By extending~\eqref{eq-commutativity} one can see that these automata are equivalent to~$\A$:
\newcommand{\stmtpropequivalence}{
 Let $\A$ be a WA.
 Then $L_\A = L_{\AF} = L_{\AB}$.
}
\begin{proposition}  \label{prop-equivalence}
 \stmtpropequivalence
\end{proposition}
Further, applying both constructions consecutively yields a minimal WA:
\newcommand{\stmtthmWAminimal}{
 Let $\A$ be a WA.
 Let $\A' = \AFB$ or $\A' = \ABF$.
 Then $\A'$ is minimal and equivalent to~$\A$.
}
\begin{theorem} \label{thm-WA-minimal}
 \stmtthmWAminimal
\end{theorem}
Theorem~\ref{thm-WA-minimal} mirrors Brzozowski's NFA minimisation algorithm.
We give a short proof in \iftechrep{Appendix~\ref{app-brzozowski}}{\cite{KW14-icalp-report}}.

\subsection{Numerically Stable WA Minimisation} \label{sub-arnoldi}

Theorem~\ref{thm-WA-minimal} reduces the problem of minimising a WA to the problem of computing a forward and a backward reduction.
In the following we focus on computing a \emph{canonical} (see above for the definition) forward reduction $(F, \MF)$.
Figure~\ref{fig-meta-algorithm} shows a generalisation of Arnoldi's iteration~\cite{Arnoldi51} to multiple matrices.
Arnoldi's iteration is typically used for locating eigenvalues~\cite{Golub}.
Its generalisation to multiple matrices is novel, to the best of the authors's knowledge.
Using~\eqref{eq-commutativity} one can see that it computes a canonical forward reduction
 by iteratively extending a partial orthonormal basis $\{ f_1, \ldots, f_j \}$ for the forward space~$\FV$.

\begin{figure}[h]
\vspace{-.2cm}
\begin{minipage}{\textwidth}
 function ArnoldiReduction \\
 input: $\alpha \in \R^n$; 
        $M: \Sigma \to \R^{n \times n}$ \\
 output: canonical forward reduction $(F, \MF)$ with $F \in \R^{\nF \times n}$ and $\MF: \Sigma \to \R^{\nF \times \nF}$ \\
  \ind $\ell$ := $0$; $j := 1$; $f_1 := \alpha / \norm{\alpha}$ \ (or $f_1 := - \alpha / \norm{\alpha}$) \\ 
  \ind while $\ell < j$ do \\
  \indd $\ell := \ell + 1$ \\
  \indd for $a \in \Sigma$ do \\
  \inddd if $f_\ell M(a) \not\in \langle f_1, \ldots, f_j \rangle$ \\
  \indddd $j := j+1$ \\ 
  \indddd define $f_j$ orthonormal to $f_1, \ldots, f_{j-1}$ such that \\
  \inddddd $\langle f_1, \ldots, f_{j-1}, f_\ell M(a) \rangle = \langle f_1, \ldots, f_j \rangle$ \\
  \inddd define $\MF(a)[\ell,\cdot]$ such that $f_\ell M(a) = \sum_{i=1}^j \MF(a)[\ell,i] f_i$ \\
  \inddd \hspace{30mm} and $\MF(a)[\ell,j{+}1..n] = (0,\ldots,0)$ \\
  \ind $\nF := j$; form $F \in \R^{\nF \times \nF}$ with rows $f_{1}, \ldots, f_{\nF}$ \\
  \ind return $F$ and $\MF(a)[1..\nF,1..\nF]$ for all $a \in \Sigma$
\end{minipage}
\caption{Generalised Arnoldi iteration.}
\vspace{-.2cm}
\label{fig-meta-algorithm}
\end{figure}

For efficiency, one would like to run generalised Arnoldi iteration (Figure~\ref{fig-meta-algorithm}) using floating-point arithmetic.
This leads to \emph{round-off errors}.
The check ``if $f_\ell M(a) \not\in \langle f_1, \ldots, f_j \rangle$'' is particularly problematic:
 since the vectors $f_1, \ldots, f_j$ are computed with floating-point arithmetic,
  we cannot expect that $f_\ell M(a)$ lies \emph{exactly} in the vector space spanned by those vectors,
   even if that would be the case without round-off errors.
As a consequence, we need to introduce an \emph{error tolerance parameter} $\tau > 0$,
 so that the check ``$f_\ell M(a) \not\in \langle f_1, \ldots, f_j \rangle$'' returns \emph{true} only
  if $f_\ell M(a)$ has a ``distance'' of more than~$\tau$ to the vector space $\langle f_1, \ldots, f_j \rangle$.%
\footnote{This will be made formal in our algorithm.}
Without such a ``fuzzy'' comparison the resulting automaton could even have more states than the original one.
The error tolerance parameter~$\tau$ causes further errors.

To assess the impact of those errors, we use the \emph{standard model of floating-point arithmetic},
 which assumes that the elementary operations $\mathord{+}, \mathord{-}, \mathord{\cdot}, \mathord{/}$ are computed exactly,
  up to a relative error of at most the \emph{machine epsilon}~$\mache \ge 0$.
It is stated in~\cite[Chapter~2]{Higham}: ``This model is valid for most computers, and, in particular, holds for IEEE standard arithmetic.''
The bit length of numbers arising in a numerical computation is bounded by hardware, using suitable roundoff.
So we adopt the convention of numerical linear algebra to take the number of arithmetic operations as a measure of time complexity.

The algorithm ArnoldiReduction (Figure~\ref{fig-meta-algorithm}) leaves open how to implement the conditional
 ``if $f_\ell M(a) \not\in \langle f_1, \ldots, f_j \rangle$'',
  and how to compute the new basis element~$f_j$.
In \iftechrep{Appendix~\ref{app-refinement}}{\cite{KW14-icalp-report}} we propose an instantiation \emph{HouseholderReduction} of ArnoldiReduction
 based on so-called \emph{Householder reflectors}~\cite{Householder58},
  which are special orthonormal matrices.
We prove the following stability property:
\newcommand{\stmtpropDeltabound}{
\begin{itemize}
\item[1.] The number of arithmetic operations is $O(|\Sigma| n^3)$.
\item[2.] HouseholderReduction instantiates ArnoldiReduction.
\item[3.] The computed matrices satisfy the following error bound:
 For each $a \in \Sigma$, the matrix $\E(a) \in \R^{\nF \times n}$ with $\E(a) := F M(a) - \MF(a) F$ satisfies
  \[
    \norm{\E(a)} \le 2 \sqrt{n} \tau + c m n^{3} \mache\;,
  \]
  where $m > 0$ is such that $\norm{M(a)} \le m$ holds for all $a \in \Sigma$, and $c>0$ is an input-independent constant.
\end{itemize}
}
\begin{proposition} \label{prop-Delta-bound}
 Consider the algorithm \emph{HouseholderReduction} in \iftechrep{Appendix~\ref{app-refinement}}{\cite{KW14-icalp-report}}, which has the following interface:
\begin{center} \upshape \small
  \begin{minipage}{\textwidth}
 function HouseholderReduction \\
 input: $\alpha \in \R^n$; 
        $M: \Sigma \to \R^{n \times n}$;
        error tolerance parameter $\tau \ge 0$ \\
 output: canonical forward reduction $(F, \MF)$ with $F \in \R^{\nF \times n}$ and $\MF: \Sigma \to \R^{\nF \times \nF}$
\end{minipage}
\end{center}
We have:
\stmtpropDeltabound
\end{proposition}
The proof follows classical error-analysis techniques for QR factorisations with Householder reflectors~\cite[Chapter~19]{Higham},
 but is substantially complicated by the presence of the ``if'' conditional and the resulting need for the $\tau$ parameter.
By Proposition~\ref{prop-Delta-bound}.2.\ HouseholderReduction computes a precise canonical forward reduction for $\mache = \tau = 0$.
For positive $\mache$ and~$\tau$ the error bound grows linearly in $\mache$ and~$\tau$, and with modest polynomials in the WA size~$n$.
In practice $\mache$ is very small%
\footnote{With IEEE double precision, e.g., it holds $\mache = 2^{-53}$~\cite{Higham}.}%
, so that the term $c m n^{3} \mache$ can virtually be ignored.

The use of Householder reflectors is crucial to obtain the bound of Proposition~\ref{prop-Delta-bound}.
Let us mention a few alternative techniques, which have been used for computing certain matrix factorisations.
Such factorisations (QR or LU) are related to our algorithm.
\emph{Gaussian elimination} can also be used for WA minimisation in time $O(|\Sigma| n^3)$,
 but its stability is governed by the \emph{growth factor}, which can be exponential even with \emph{pivoting}~\cite[Chapter~9]{Higham},
  so the bound on $\norm{\E(a)}$ in Proposition~\ref{prop-Delta-bound} would include a term of the form $2^n \mache$.
The most straightforward implementation of ArnoldiReduction would use the \emph{Classical Gram-Schmidt} process,
 which is highly unstable~\cite[Chapter~19.8]{Higham}.
A variant, the \emph{Modified Gram-Schmidt} process is stable,
 but the error analysis is complicated by a possibly loss of orthogonality of the computed matrix~$F$.
The extent of that loss depends on certain condition numbers (cf.~\cite[Equation~(19.30)]{Higham}), which are hard to estimate or control in our case.
In contrast, our error bound is independent of condition numbers.

Using Theorem~\ref{thm-WA-minimal} we can prove:
\newcommand{\stmtthmerrorbound}{
Consider the following algorithm:
\begin{center} \upshape \small
\begin{minipage}{\textwidth}
 function HouseholderMinimisation \\
 input: WA $\A = (n, \Sigma, M, \alpha, \eta)$;
        error tolerance parameter $\tau \ge 0$ \\
 output: minimised WA $\A' = (n', \Sigma, M', \alpha', \eta')$. \\
 \ind compute forward reduction $(F, \MF)$ of~$\A$ using HouseholderReduction \\
 \ind form $\AF := (\nF, \Sigma, \MF, \alphaF, \etaF)$ as the forward WA of~$\A$ with base~$F$ \\
 \ind compute backward reduction $(B, M')$ of~$\AF$ using HouseholderReduction \\
 \ind form $\A' := (n', \Sigma, M', \alpha', \eta')$ as the backward WA of~$\AF$ with base~$B$ \\
 \ind return $\A'$
\end{minipage}
\end{center}
We have:
 \begin{itemize}
  \item[1.] The number of arithmetic operations is $O(|\Sigma| n^3)$.
  \item[2.] For $\mache = \tau = 0$, the computed WA $\A'$ is minimal and equivalent to~$\A$.
  \item[3.] Let $\tau > 0$.
   Let $m > 0$ such that $\norm{A} \le m$ holds for all \hbox{$A \in \{M(a), \MF(a), M'(a) \mid a \in \Sigma\}$}.
   Then for all $w \in \Sigma^*$ we have
   \begin{align*}
    |L_\A(w) - L_{\A'}(w)| & \le 4 |w| \norm{\alpha} m^{|w|-1} \norm{\eta} \sqrt{n} \tau \\
                           & \quad \mbox{} + c \max\{|w|,1\} \norm{\alpha} m^{|w|} \norm{\eta} n^3 \mache\;,
   \end{align*}
   where $c>0$ is an input-independent constant.
 \end{itemize}
}
\begin{theorem} \label{thm-error-bound}
\stmtthmerrorbound
\end{theorem}
The algorithm computes a backward reduction by running the straightforward backward variant of HouseholderReduction.
We remark that for PAs one can take $m = 1$ for the norm bound~$m$ from part~3.\ of the theorem
 (or $m = 1 + \varepsilon$ for a small $\varepsilon$ if unfortunate roundoff errors occur).
It is hard to avoid an error bound exponential in the word length~$|w|$,
 as $|L_\A(w)|$ itself may be exponential in~$|w|$ (consider a WA of size~$1$ with $M(a) = 2$).
Theorem~\ref{thm-error-bound} is proved in \iftechrep{Appendix~\ref{app-error}}{\cite{KW14-icalp-report}}.

The error bounds in Proposition~\ref{prop-Delta-bound} and Theorem~\ref{thm-error-bound}
  suggest to choose a small value for the error tolerance parameter~$\tau$.
But as we have discussed, the computed WA may be non-minimal if $\tau$ is set too small or even to~$0$,
 intuitively because round-off errors may cause the algorithm to overlook minimisation opportunities.
So it seems advisable to choose $\tau$ smaller (by a few orders of magnitude) than the desired bound on~$\norm{\E(a)}$,
 but larger (by a few orders of magnitude) than~$\mache$.
Note that for $\mache > 0$ Theorem~\ref{thm-error-bound}
 does not provide a bound on the number of states of~$\A'$.

To illustrate the stability issue we have experimented with minimising a PA~$\A$ derived from Herman's protocol as in~\cite{KMOWW:CAV11}.
The PA has 190 states and $\Sigma = \{a\}$.
When minimising with the (unstable) Classical Gram-Schmidt process,
 we have measured a huge error of $|L_\A(a^{190}) - L_{\A'}(a^{190})| \approx 10^{36}$.
With the Modified Gram-Schmidt process and the method from Theorem~\ref{thm-error-bound}
 the corresponding errors were about $10^{-7}$, which is in the same order as the error tolerance parameter~$\tau$.

\subsection{Lossy WA Minimisation} \label{sub-min-image}

A larger error tolerance parameter~$\tau$ leads to more ``aggressive'' minimisation of a possibly already minimal WA.
The price to pay is a shift in the language: one would expect only $L_\A'(w) \approx L_{\A}(w)$.
Theorem~\ref{thm-error-bound} provides a bound on this imprecision.
In this section we illustrate the trade-off between size and precision using an application in image compression.

Weighted automata can be used for image compression, as suggested by Culik et al.~\cite{CulikKari93}.
An image, represented as a two-dimensional matrix of grey-scale values,
can be encoded as a weighted automaton where each pixel is addressed by a unique word.
To obtain this automaton, the image is recursively subdivided into quadrants.
There is a state for each quadrant and transitions from a quadrant to its
sub-quadrants. At the level of the pixels, the automaton accepts
with the correct grey-scale value.

Following this idea, we have implemented a prototype tool for image compression based on the algorithm of Theorem~\ref{thm-error-bound}.
We give details and show example pictures in \iftechrep{Appendix~\ref{app-image-compression}}{\cite{KW14-icalp-report}}.
This application illustrates lossy minimisation.
The point is that Theorem~\ref{thm-error-bound} guarantees bounds on the loss.

\section{The Complexity of PA Minimisation} \label{sec-minpro}

Given a PA $\A = (n, \Sigma, M, \alpha, \eta)$ and $n' \in \N$,
 the \emph{PA minimisation problem} asks whether there exists a PA $\A' = (n', \Sigma, M', \alpha', \eta')$
 so that $\A$ and~$\A'$ are equivalent.
For the complexity results in this section we assume that the numbers in the description of the given PA
 are fractions of natural numbers represented in binary, so they are rational.
In Section~\ref{sub-NPhardness} we show that the minimisation problem is NP-hard.
In Section~\ref{sub-PSPACE} we show that the problem is in PSPACE
 by providing a polynomial-time reduction to the existential theory of the reals.

\subsection{NP-Hardness} \label{sub-NPhardness}

We will show:
\begin{theorem} \label{thm-NPhardness}
The PA minimisation problem is NP-hard.
\end{theorem}

For the proof we reduce from a geometrical problem, the \emph{hypercube problem}, which we show to be NP-hard.
Given $d \in \N$, a finite set $P = \{p_1, \ldots, p_k\} \subseteq [0,1]^d$ of vectors (``points'') within the $d$-dimensional unit hypercube, and $\ell \in \N$,
 the \emph{hypercube problem} asks whether there is a set $Q = \{q_1, \ldots, q_\ell\} \subseteq [0,1]^d$ of at most $\ell$ points within the hypercube
  such that $\conv(Q) \supseteq P$, where
   \[
    \conv(Q) := \{\lambda_1 q_1 + \cdots + \lambda_\ell q_\ell \mid \lambda_1, \ldots, \lambda_\ell \ge 0, \ \lambda_1 + \cdots + \lambda_\ell = 1\}
   \]
    denotes the convex hull of~$Q$.
Geometrically, the convex hull of~$P$ can be viewed as a convex polytope, nested inside the hypercube, which is another convex polytope.
The hypercube problem asks whether a convex polytope with at most $\ell$ vertices can be nested in between those polytopes.
The answer is trivially yes, if $\ell \ge k$ (take $Q = P$) or if $\ell \ge 2^d$ (take $Q = \{0,1\}^d$).
We speak of the \emph{restricted hypercube problem} if $P$ contains the origin $(0, \ldots, 0)$.
We prove the following:
\newcommand{\stmtpropredCubeAut}{
 The restricted hypercube problem can in polynomial time be reduced to the PA minimisation problem.
}
\begin{proposition} \label{prop-redCubeAut}
 \stmtpropredCubeAut
\end{proposition}
\begin{proof}[sketch]
Let $d \in \N$ and $P = \{p_1, \ldots, p_{k}\} \subseteq [0,1]^d$ and $\ell \in \N$ be an instance of the restricted hypercube problem,
 where $p_1 = (0, \ldots, 0)$ and $\ell \ge 1$.
We construct in polynomial time a PA $\A = (k+1, \Sigma, M, \alpha, \eta)$ such that
 there is a set $Q = \{q_1, \ldots, q_\ell\} \subseteq [0,1]^d$ with $\conv(Q) \supseteq P$
  if and only if
 there is a PA $\A' = (\ell+1, \Sigma, M', \alpha', \eta')$ equivalent to~$\A$.
Take $\Sigma := \{a_2, \ldots, a_{k}\} \cup \{b_1, \ldots, b_d\}$.
Set
 $M(a_i)[1, i] := 1$ and
 $M(b_s)[i, k+1] := p_i[s]$ for all $i \in \{2, \ldots, k\}$ and all $s \in \set{d}$,
  and set all other entries of~$M$ to~$0$.
Set $\alpha := e(1)$ and $\eta := e(k+1)^T$.
Figure~\ref{fig-redCubeAut} shows an example of this reduction.
We prove the correctness of this reduction in \iftechrep{Appendix~\ref{app-prop-redCubeAut}}{\cite{KW14-icalp-report}}.
\begin{figure}
\begin{tikzpicture}[scale=2.3,baseline=(middle)]
\coordinate (00) at (0,0);
\coordinate (01) at (0,1);
\coordinate (10) at (1,0);
\coordinate (11) at (1,1);
\coordinate (middle) at (0.5,0.5);
\draw (00) -- (01) -- (11) -- (10) -- (00);
\draw[fill] (00) circle (0.03);
\node at (0.08,0.1) {$p_1$};
\draw[fill] (0,3/4) circle (0.03);
\node at (0.1,0.75) {$p_2$};
\draw[fill] (1/4,1/2) circle (0.03);
\node at (0.35,0.5) {$p_3$};
\draw[fill] (1/2,1/4) circle (0.03);
\node at (0.6,0.2) {$p_4$};
\draw[fill] (1/2,3/4) circle (0.03);
\node at (0.6,0.8) {$p_5$};
\draw[thick] (0,0) -- (1,1/2) -- (0,1) -- (0,0);
\end{tikzpicture}
\hfill
\begin{tikzpicture}[scale=2,baseline,>=stealth',every state/.style={minimum size=0.3,inner sep=1}]
\node[state] (1) at (0,0) {$1$};
\node[state] (2) at (1,1) {$2$};
\node[state] (3) at (1,0.5) {$3$};
\node[state] (4) at (1,-0.5) {$4$};
\node[state] (5) at (1,-1) {$5$};
\node[state,accepting] (6) at (2,0) {$6$};
\draw[->] (-0.3,0) -- (1);
\draw[->] (1) -- node[above=2mm] {$1 a_2$} (2);
                                 \draw[->] (2) -- node[above=2mm,pos=0.6] {$\frac34 b_2$} (6);
\draw[->] (1) -- node[below] {$1 a_3$} (3);
                                 \draw[->] (3) -- node[below=-0.5mm,pos=0.25] {$\frac14 b_1$} node[below=-1mm,pos=0.55] {$\frac12 b_2$} (6);
\draw[->] (1) -- node[above] {$1 a_4$} (4);
                                 \draw[->] (4) -- node[above=-0.5mm,pos=0.25] {$\frac12 b_1$} node[above=-1mm,pos=0.55] {$\frac14 b_2$} (6);
\draw[->] (1) -- node[below=1mm] {$1 a_5$} (5);
                                 \draw[->] (5) -- node[below=1mm,pos=0.3] {$\frac12 b_1$} node[below=1mm,pos=0.6] {$\frac34 b_2$} (6);
\end{tikzpicture}
\hfill
\begin{tikzpicture}[scale=1.9,baseline,>=stealth',every state/.style={minimum size=0.3,inner sep=1}]
\node[state] (1) at (0,0) {$1$};
\node[state] (2) at (1,1) {$2$};
\node[state] (3) at (1,-1) {$3$};
\node[state,accepting] (4) at (2,0) {$4$};
\draw[->] (-0.3,0) -- (1);
\draw[->] (1) -- node[above=1mm,pos=0.1] {$\frac34 a_2$} node[above=1mm,pos=0.4] {$\frac38 a_3$} node[above=1mm,pos=0.7] {$\frac12 a_5$} (2);
                                 \draw[->] (2) --                                   node[above=2mm,pos=0.8] {$1 b_2$} (4);
\draw[->] (1) -- node[below=1mm,pos=0.1] {$\frac14 a_3$} node[below=0.5mm,pos=0.4] {$\frac12 a_4$} node[below=0.5mm,pos=0.7] {$\frac12 a_5$} (3);
                                 \draw[->] (3) -- node[below=1mm,pos=0.4] {$1 b_1$} node[below=1mm,pos=0.8] {$\frac12 b_2$} (4);
\end{tikzpicture}
\caption{Reduction from the hypercube problem to the minimisation problem.
The left figure shows an instance of the hypercube problem with $d=2$ and
 $P = \{p_1, \ldots, p_5\} = \{(0,0), (0,\frac34), (\frac14,\frac12), (\frac12,\frac14), (\frac12,\frac34)\}$.
It also suggests a set $Q = \{(0,0), (0,1), (1,\frac12)\}$ with $\conv(Q) \supseteq P$.
The middle figure depicts the PA~$\A$ obtained from~$P$.
The right figure depicts a minimal equivalent PA~$\A'$, corresponding to the set~$Q$ suggested in the left figure.
}
\label{fig-redCubeAut}
\end{figure}
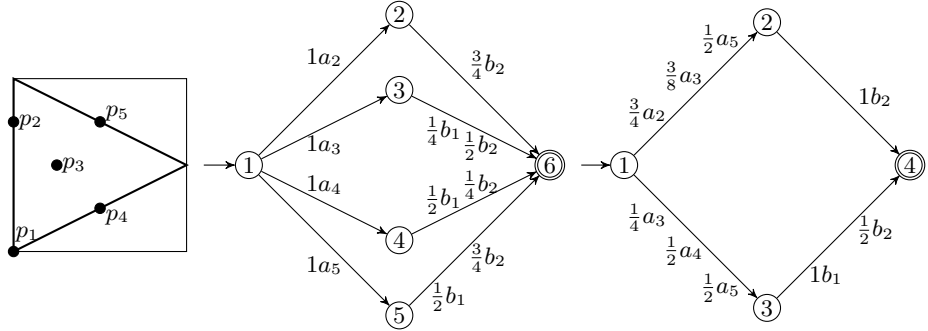
\qed
\end{proof}

Next we show that the hypercube problem is NP-hard, which together with Proposition~\ref{prop-redCubeAut} implies Theorem~\ref{thm-NPhardness}.
A related problem is known\footnote{The authors thank Joseph O'Rourke for pointing out \cite{DasGoodrich95}.} to be NP-hard:
\begin{theorem}[Theorem~4.2 of \cite{DasGoodrich95}] \label{thm-DasGoodrich}
 Given two nested convex polyhedra in three dimensions, the problem of nesting a convex polyhedron with minimum faces between the two polyhedra is NP-hard.
\end{theorem}
Note that this NP-hardness result holds even in $d=3$ dimensions.
However, the outer polyhedron is not required to be a cube, and the problem is about minimising
 the number of faces rather than the number of vertices.
Using a completely different technique we show:
\newcommand{\stmtpropcubeNPhardness}{
 The hypercube problem is NP-hard.
 This holds even for the restricted hypercube problem.
}
\begin{proposition} \label{prop-cube-NPhardness}
\stmtpropcubeNPhardness
\end{proposition}
The proof is by a reduction from 3SAT, see \iftechrep{Appendix~\ref{app-prop-cube-NPhardness}}{\cite{KW14-icalp-report}}.

\begin{remark} \label{rem-hypercube}
The hypercube problem is in PSPACE, by appealing to decision algorithms for~$\ExThR$, the existential fragment of the first-order theory of the reals.
For every fixed~$d$ the hypercube problem is\footnote{This observation is in part due to Radu Grigore.} in~$P$,
 exploiting the fact that $\ExThR$ can be decided in polynomial time, if the number of variables is fixed.
(For $d=2$ an efficient algorithm is provided in~\cite{Aggarwal89}.)
It is an open question whether the hypercube problem is in~NP.
It is also open whether the search for a minimum~$Q$ can be restricted to sets of points with rational coordinates
 (this holds for $d=2$).
\end{remark}

Propositions \ref{prop-redCubeAut} and~\ref{prop-cube-NPhardness} together imply Theorem~\ref{thm-NPhardness}.

\subsection{Reduction to the Existential Theory of the Reals} \label{sub-PSPACE}

In this section we reduce the PA minimisation problem to $\ExThR$, the existential fragment of the first-order theory of the reals.
A formula of~$\ExThR$ is of the form
 $\exists x_1 \ldots \exists x_m R(x_1, \ldots, x_n)$, where $R(x_1, \ldots, x_n)$ is a boolean combination of comparisons of the form
 $p(x_1, \ldots, x_n) \sim 0$, where $p(x_1, \ldots, x_n)$ is a multivariate polynomial and
  $\mathord{\sim} \in \{ \mathord{<}, \mathord{>}, \mathord{\le}, \mathord{\ge}, \mathord{=}, \mathord{\ne} \}$.
The validity of closed formulas ($m=n$) is decidable in PSPACE~\cite{Can88,Renegar92},
 and is not known to be PSPACE-hard.

\newcommand{\stmtpropPSPACE}{
 Let $\A_1 = (n_1, \Sigma, M_1, \alpha_1, \eta_1)$ be a PA.
 A PA \hbox{$\A_2 = (n_2, \Sigma, M_2, \alpha_2, \eta_2)$} is equivalent to~$\A_1$ if and only if
  there exist matrices $\MF(a) \in \R^{(n_1+n_2) \times (n_1+n_2)}$ for $a \in \Sigma$
  and a matrix $F \in \R^{(n_1+n_2) \times (n_1+n_2)}$ such that
  $F[1,\cdot] = (\alpha_1, \alpha_2)$, and $F (\eta_1^T, - \eta_2^T)^T = (0, \ldots, 0)^T$,
  and
  \[
   F \begin{pmatrix} M_1(a) & 0 \\ 0 & M_2(a) \end{pmatrix} = \MF(a) F \qquad \text{for all $a \in \Sigma$.}
  \]
}
\begin{proposition} \label{prop-PSPACE}
\stmtpropPSPACE
\end{proposition}
The proof is in \iftechrep{Appendix~\ref{app-prop-PSPACE}}{\cite{KW14-icalp-report}}.
The conditions of Proposition~\ref{prop-PSPACE} on~$\A_2$, including that it be a PA, can be phrased in~$\ExThR$.
Thus it follows:
\begin{theorem} \label{thm-PSPACE}
 The PA minimisation problem can be reduced in polynomial time to~$\ExThR$.
 Hence, PA minimisation is in PSPACE.
\end{theorem}


Theorem~\ref{thm-PSPACE} improves on a result in~\cite{MateusQL12} where the minimisation problem was shown to be in EXPTIME.
(More precisely, Theorem~4 of~\cite{MateusQL12} states that a minimal PA can be computed in EXPSPACE,
 but the proof reveals that the decision problem can be solved in EXPTIME.)

\section{Conclusions and Open Questions} \label{sec-conclusions}

We have developed a numerically stable and efficient algorithm for minimising WAs, based on linear algebra
 and Brzozowski-like automata minimisation.
We have given bounds on the minimisation error in terms of both the machine epsilon and the error tolerance parameter~$\tau$.

We have shown NP-hardness for PA minimisation, and have given a polynomial-time reduction to~$\ExThR$.
Our work leaves open the precise complexity of the PA minimisation problem.
The authors do not know whether the search for a minimal PA can be restricted to PAs with rational numbers.
As stated in the Remark after Proposition~\ref{prop-cube-NPhardness}, the corresponding question is open even for the hypercube problem.
If rational numbers indeed suffice, then an NP algorithm might exist that guesses the (rational numbers of the) minimal PA and checks for equivalence
 with the given PA.
Proving PSPACE-hardness would imply PSPACE-hardness of~$\ExThR$, thus solving a longstanding open problem.

For comparison, the corresponding minimisation problems involving WAs (a generalisation of PAs)
 and DFAs (a special case of PAs) lie in~$P$.
More precisely, minimisation of WAs (with rational numbers) is in randomised NC~\cite{13KMOWW-LMCS},
 and DFA minimisation is NL-complete~\cite{ChoH92}.
NFA minimisation is PSPACE-complete~\cite{Jiang93}.

\medskip
\noindent \textbf{Acknowledgements.} The authors would like to thank James Worrell, Radu Grigore, and Joseph O'Rourke for valuable discussions,
and the anonymous referees for their helpful comments.
Stefan Kiefer is supported by a Royal Society University Research Fellowship.

\bibliographystyle{plain}
\bibliography{db}

\iftechrep{
\newpage
\appendix
\section{Proofs of Section~\ref{sec-minwei}}

\subsection{Proof of Proposition~\ref{prop-equivalence}}

\begin{qproposition}{\ref{prop-equivalence}}
 \stmtpropequivalence
\end{qproposition}
\begin{proof}
 Observe that the equalities~\eqref{eq-commutativity} extend inductively to words:
 \begin{equation}
  F M(w) = \MF(w) F \quad \text{and} \quad M(w) B = B \MB(w) \quad \text{for all $w \in \Sigma^*$.}
   \label{eq-commutativity-words}
 \end{equation}
 Using \eqref{eq-commutativity-words} and the definition of~$\alphaF$ we have for all $w \in \Sigma^*$:
 \[
     L_\A(w)
     = \alpha M(w) \eta
     = \alphaF F M(w) \eta
     = \alphaF \MF(w) F \eta
     = L_{\AF}(w)\,.
 \]
 Symmetrically one can show $L_\A = L_{\AB}$.
\qed
\end{proof}

\subsection{Proof of Theorem~\ref{thm-WA-minimal}} \label{app-brzozowski}

We will use the notion of a \emph{Hankel matrix} \cite{BerstelReutenauer,Beimel00}:
\begin{definition}
 Let $L : \Sigma^* \to \R$.
 The \emph{Hankel matrix} of~$L$ is the matrix \mbox{$H^L \in \R^{\Sigma^* \times \Sigma^*}$}
  with $H^L[x,y] = L(x y)$ for all $x,y \in \Sigma^*$.
 We define $\rank(L) := \rank(H^L)$.
\end{definition}
We have the following proposition:
\begin{proposition} \label{prop-direction-1}
 Let $\A$ be an automaton of size~$n$. Then $\rank(L_\A) \le n$.
\end{proposition}
\begin{proof}
 Consider the matrices $\Fh: \R^{\Sigma^* \times n}$ and $\Bh: \R^{n \times \Sigma^*}$ with
  $\Fh[w,\cdot] := \alpha M(w)$ and $\Bh[\cdot,w] := M(w) \eta$ for all $w \in \Sigma^*$.
 Note that $\rank(\Fh) \le n$ and $\rank(\Bh) \le n$.
 Let $x,y \in \Sigma^*$.
 Then $(\Fh \Bh)[x,y] = \alpha M(x) M(y) \eta = L_\A(x y)$, so $\Fh \Bh$ is the Hankel matrix of~$L_\A$.
 Hence $\rank(L_\A) = \rank(\Fh \Bh) \le \min\{\rank(\Fh), \rank(\Bh)\} \le n$.
\qed
\end{proof}
Now we can prove the theorem:
\begin{qtheorem}{\ref{thm-WA-minimal}}
 \stmtthmWAminimal
\end{qtheorem}
\begin{proof}
W.l.o.g.\ we assume $\A' = \ABF$.
Let $\A = (n, \Sigma, M, \alpha, \eta)$.
Let $\AB = (\nB, \Sigma, \MB, \alpha B, \etaB)$ be a backward automaton of~$\A$ with base~$B$.
Let $\ABF$ 
 be a forward automaton of~$\AB$ with base~$\Ft$.
Equivalence of $\A$, $\AB$ and~$\ABF$ follows from Proposition~\ref{prop-equivalence}.
Assume that $\ABF$ has $\nBF$ states.
For minimality, by Proposition~\ref{prop-direction-1}, it suffices to show $\nBF = \rank(H)$, where $H$ is the Hankel matrix of~$L_\A$.
Let $\Fh$ and $\Bh$ be the matrices from the proof of Proposition~\ref{prop-direction-1}.
We have:
\begin{align*}
  \nBF
   & = \rank (\Ft) && \text{(definition of~$\nBF$)} \\
   & = \dim \langle \alpha B  \MB(w) \mid w \in \Sigma^* \rangle && \text{(definition of~$\Ft$)} \\
   & = \dim \langle \alpha M(w) B \mid w \in \Sigma^* \rangle && \text{(by~\eqref{eq-commutativity-words})} \\
   & = \rank(\Fh B) && \text{(definition of~$\Fh$)} \\
   & = \dim \langle \Fh M(w) \eta \mid w \in \Sigma^* \rangle && \text{(definition of~$B$)} \\
   & = \rank(\Fh \Bh) && \text{(definition of~$\Bh$)} \\
   & = \rank(H) && \text{(proof of Proposition~\ref{prop-direction-1})\,.}
\end{align*}
\qed
\end{proof}

\subsection{Instantiation of ArnoldiReduction with Householder Reflectors} \label{app-refinement}

Fix $n \ge 1$.
For a row vector $x \in \R^k$ with $k \in \N_n$ and $\norm{x} \ne 0$, the \emph{Householder reflector $P$ for~$x$} is defined
 as the matrix
\[
 P = \begin{pmatrix} I_{n-k} & \ 0 \\ 0 & \ R \end{pmatrix} \in \R^{n \times n}\;,
\]
 where $R = I_k - 2 v^T v \in \R^{k \times k}$, and
\[
 v = \frac{(x[1] + \sign(x[1]) \norm{x}, x[2], \ldots, x[k])}{\norm{(x[1] + \sign(x[1]) \norm{x}, x[2], \ldots, x[k])}} \in \R^k \;,
\]
 where $\sign(r) = +1$ if $r \ge 0$ and $-1$ otherwise.
 (The careful choice of the sign here is to ensure numerical stability.)
To understand this definition better, first observe that $v$ is a row vector with $\norm{v} = 1$.
It is easy to verify that $R$ and thus $P$ are orthonormal and symmetric.
Moreover, we have $R R = I_k$ and thus $P P = I_n$, i.e., $P = P^T = P^{-1}$.
Geometrically, $R$ describes a reflection about the hyperplane through the origin and orthogonal to $v \in \R^{k}$~\cite[Chapter~19.1]{Higham}.
Crucially, the vector~$v$ is designed so that $R$ reflects~$x$ onto the first axis, i.e., $x R = (\pm \norm{x}, 0, \ldots, 0)$.

\begin{figure}[h]
\begin{minipage}{\textwidth}
 function HouseholderReduction \\
 input: $\alpha \in \R^n$; 
        $M: \Sigma \to \R^{n \times n}$;
        error tolerance parameter $\tau \ge 0$ \\
 output: canonical forward reduction $(F, \MF)$ with $F \in \R^{\nF \times n}$ and $\MF: \Sigma \to \R^{\nF \times \nF}$ \\
 \ind $\Ps{1}$ := Householder reflector for~$\alpha  / \norm{\alpha}$ \\
 \ind $\ell$ := $0$; $j := 1$; $f_1 := e(1) \Ps{1}$ \\ 
 \ind while $\ell < j$ do \\
 \indd $\ell := \ell + 1$ \\
 \indd for $a \in \Sigma$ do \\
 \inddd $\MF(a)[\ell, \cdot] := f_\ell M(a) \Ps{1} \cdots \Ps{j}$ \\
 \inddd if $j + 1 \le n$ and $\norm{\MF(a)[\ell,j{+}1..n]} > \tau$ \\
 \indddd $j := j+1$ \\ 
 \indddd $\Ps{j}$ := Householder reflector for $\MF(a)[\ell,j..n]$ \\
 \indddd $\MF(a)[\ell, \cdot] := \MF(a)[\ell, \cdot] \Ps{j}$ \\
 \indddd $f_j := e(j) \Ps{j} \cdots \Ps{1}$\\
 \ind $\nF := j$; form $F \in \R^{\nF \times n}$ with rows $f_{1}, \ldots, f_{\nF}$ \\
 \ind return $F$ and $\MF(a)[1..\nF,1..\nF]$ for all $a \in \Sigma$
\end{minipage}
\caption{Instantiation of ArnoldiReduction (Fig.~\ref{fig-meta-algorithm}) using Householder reflectors.}
\label{fig-householder-tzeng}
\end{figure}

Figure~\ref{fig-householder-tzeng} shows an instantiation of the algorithm from Figure~\ref{fig-meta-algorithm} using Householder reflectors.
Below we prove correctness by showing that HouseholderReduction indeed refines ArnoldiReduction,
 assuming $\tau = 0$ for the error tolerance parameter.
For efficiency it is important not to form the reflectors $P_1, P_2, \ldots$ explicitly.
If $P$ is the Householder reflector for $x \in \R^k$, it suffices to keep the vector $v \in \R^k$ from the definition of Householder reflectors.
A multiplication $y := y P$ (for a row vector $y \in \R^n$) can then be implemented in $O(n)$ with
$
 y[(n{-}k{+}1)..n] := y[(n{-}k{+}1)..n] - 2 \left( y[(n{-}k{+}1)..n] \cdot v^T \right) v
$.
This gives an $O(|\Sigma| n^3)$ number of arithmetic operations of HouseholderReduction.

\subsubsection{HouseholderReduction Instantiates ArnoldiReduction.}
Let $\Ps{1}, \ldots, \Ps{\nF} \in \R^{n \times n}$ be the reflectors computed in HouseholderReduction.
Recall that we have $\Ps{j}^T = \Ps{j}^{-1} = \Ps{j}$.
For $0 \le j \le \nF$ define
\[
 \Fs{j} := \Ps{j} \Ps{j-1} \cdots \Ps{1}\,. 
\]
The matrices $\Fs{j} \in \R^{n \times n}$ are orthonormal. 
Since for all $j < \nF$ the reflector $\Ps{j+1}$ leaves the first $j$ rows of~$\Fs{j}$ unchanged,
 the first $j$ rows of $\Fs{j}, \ldots, \Fs{\nF}$ coincide with the vectors $f_1, \ldots, f_j$ computed in HouseholderReduction.

First we consider the initialisation part before the while loop.
Since $\Ps{1}$ is the Householder reflector for~$\alpha/\norm{\alpha}$, we have
 $(\alpha / \norm{\alpha}) \Ps{1} = (\pm 1, 0, \ldots, 0)$.
It follows that we have $f_1 = e(1) \Ps{1} = \pm \alpha / \norm{\alpha}$, as in ArnoldiReduction.

Now we consider the while loop.
Consider an arbitrary iteration of the ``for $a \in \Sigma$ do'' loop.
Directly after the loop head we have $\MF(a)[\ell, \cdot] = f_\ell M(a) \Ps{1} \cdots \Ps{j}$,
 hence $f_\ell M(a) = \MF(a)[\ell, \cdot] \Fs{j}$.
As $\Fs{j}$ is orthonormal and its first $j$ rows coincide with $f_1, \ldots, f_j$,
 we have $f_\ell M(a) \in \langle f_1, \ldots, f_j \rangle$ if and only if $\MF(a)[\ell,j{+}1 .. n] = (0, \ldots, 0)$.
This corresponds to the ``if'' conditional in ArnoldiReduction.
It remains to be shown that $\MF(a)[\ell, \cdot]$ is defined as in ArnoldiReduction:
\begin{itemize}
\item
Let $\MF(a)[\ell, j{+}1 .. n] = (0, \ldots, 0)$.
Then at the end of the loop we have $f_\ell M(a) = \sum_{i=1}^j \MF(a)[\ell,i] f_i$,
 as required in ArnoldiReduction.
\item
Let $\MF(a)[\ell,j{+}1 .. n] \ne (0, \ldots, 0)$.
Then we have:
\begin{align*}
 f_\ell M(a)
 & = \MF(a)[\ell, \cdot] \Fs{j} && \text{(as argued above)} \\
 & = \MF(a)[\ell, \cdot] \Ps{j+1} \Ps{j+1} \Fs{j} && \text{(as $\Ps{j+1} \Ps{j+1} = I_n$)} \\
 & = \MF(a)[\ell, \cdot] \Ps{j+1} \Fs{j+1} && \text{(by definition of $\Fs{j+1}$)}
\end{align*}
Further, the reflector~$\Ps{j+1}$ is designed such that $\MF(a)[\ell,j{+}1..n] = (r, 0, \ldots, 0)$ for some $r \ne 0$.
So after increasing~$j$, and updating~$\MF(a)[\ell, \cdot]$, and defining the new vector $f_j$,
 at the end of the loop we have $f_\ell M(a) = \sum_{i=1}^j \MF(a)[\ell,i] f_i$,
 as required in ArnoldiReduction.
\end{itemize}

\subsection{Proof of Proposition~\ref{prop-Delta-bound}}

\begin{qproposition}{\ref{prop-Delta-bound}}
 Consider HouseholderReduction (Figure~\ref{fig-householder-tzeng}).
 We have:
 \stmtpropDeltabound
\end{qproposition}
\begin{proof}
The fact that HouseholderReduction is an instance of ArnoldiReduction was proved in the previous subsection.
There we also showed the bound on the number of arithmetic operations.
It remains to prove the error bound.
In order to highlight the gist of the argument we consider first the case $\mache = 0$.
For $0 \le j \le n$ we write $F_j := P_j P_{j-1} \cdots P_1$.
Recall that $f_1, \ldots, f_j$, i.e., the first $j$ rows of~$F$, coincide with the first $j$ rows of~$F_j$.
We have at the end of the for loop:
\begin{align}
 f_\ell M(a) & = \MF(a)[\ell, 1..n] F_j \notag \\
             & = \MF(a)[\ell, 1..\nF] F  - \MF(a)[\ell, j{+}1..\nF] F[j{+}1..\nF,\cdot] \label{eq-slices} \\
             & \hspace{29mm} + \MF(a)[\ell, j{+}1..n] F_j[j{+}1..n,\cdot] \notag
\end{align}
So we have, for $\ell \in \N_{\nF}$:
 \[
  \E(a)[\ell, \cdot] = - \MF(a)[\ell, j{+}1..\nF] F[j{+}1..\nF,\cdot] \ + \ \MF(a)[\ell, j{+}1..n] F_j[j{+}1..n,\cdot]
 \]
Thus:
\begin{equation}
\begin{aligned}
 & \norm{\E(a)[\ell, \cdot]} \\
 & = \norm{- \MF(a)[\ell, j{+}1..\nF] F[j{+}1..\nF,\cdot] \ + \ \MF(a)[\ell, j{+}1..n] F_j[j{+}1..n,\cdot]} \\
 & \le \norm{\MF(a)[\ell, j{+}1..\nF]} + \norm{\MF(a)[\ell, j{+}1..n]} \\
 & \le 2 \norm{\MF(a)[\ell, j{+}1..n]} \\
 & \le 2 \tau \;,
\end{aligned}
\label{eq-Delta-bound}
\end{equation}
where the first inequality is by the fact that the rows of $F$ and $F_j$ are orthonormal,
 and the last inequality by the ``if'' conditional in HouseholderReduction.

From the row-wise bound~\eqref{eq-Delta-bound} we get the following bound on the matrix norm, see \cite[Lemma~6.6.a]{Higham}:
\[
 \norm{\E(a)}_2 \le 2 \sqrt{n} \tau
\]

We now consider $\mache > 0$.
We perform an error analysis similar to the one in~\cite[Chapter 19]{Higham} for QR factorisations with Householder reflectors.
Our situation is complicated by the error tolerance parameter~$\tau$;
 i.e., we have to handle a combination of the errors caused by~$\tau$ (as analysed above) and numerical errors.
In the following we distinguish between \emph{computed} quantities (with numerical errors and indicated with a ``hat'' accent)
 and \emph{ideal} quantities (without numerical errors, no ``hat'' accent).
It is important to note that when we speak of ideal quantities,
 we assume that we perform the \emph{same} arithmetic operations (multiplications, additions, etc.) as in the computed case;
 i.e., the only difference between computed and ideal quantities is that the computed quantities come with numerical errors.
In particular, we assume that the boolean values that the ``if'' conditional in HouseholderReduction evaluates to are the \emph{same}
 for the ideal computation.
The error caused by ``wrong'' evaluations of the conditional are already captured in the analysis above.
For the remaining analysis it suffices to add the (purely) numerical error.

Concretely, we write $\hatF \in \R^{\nF \times n}$ for the computed version of~$F$, and $\hatf_\ell$ for its rows.
For $a \in \Sigma$ we write $\MF(a) \in \R^{\nF \times n}$ for the \emph{computed} quantity,
 as we do not consider an ideal version (and to avoid clutter).
So we wish to bound the norm of $\E(a) := \hatF M(a) - \MF(a)[\cdot, 1..\nF] \hatF$.
By observing that $\hatF$ arises by (numerically) multiplying (at most~$n$) Householder reflectors,
 and by invoking the analysis from~\cite[Chapter~19.3]{Higham}
 (specifically the computation leading to Equation~(19.13)) we have
\begin{equation} \label{eq-error-F}
 \hatF = F + \Oti(n^{2.5} \mache) \;,
\end{equation}
where by $\Oti(p(n) \mache)$ for a polynomial~$p$ we mean a matrix or vector~$A$ of appropriate dimension with $\norm{A} \le c p(n) \mache$
 for a constant~$c > 0$.%
\footnote{We do not ``hunt down'' the constant~$c$. The analysis in~\cite[Chapter~19.3]{Higham} is similar-
 There the analogous constants are not computed either but are called ``small''.}
(In~\eqref{eq-error-F} we would have $A \in \R^{\nF \times n}$ with $\norm{A} \le c n^{2.5} \mache$.)
For $a \in \Sigma$ and $\ell \in \N_{\nF}$ define:
\begin{equation} \label{eq-error-def-g}
 g_{\ell, a} := \hatf_\ell M(a)
\end{equation}
Let $\hatg_{\ell, a}$ be the corresponding computed quantity.
Then we have, see~\cite[Chapter~3.5]{Higham}:
\begin{equation} \label{eq-error-g}
 g_{\ell, a} = \hatg_{\ell, a} + \Oti(m n^{1.5} \mache)
\end{equation}
Observe from the algorithm that $\MF(a)[\ell, \cdot]$ is computed by applying at most~$n$ Householder reflectors to~$\hatg_{\ell, a}$.
It follows with~\cite[Lemma~19.3]{Higham}:
\begin{equation} \label{eq-error-ghat1}
 \MF(a)[\ell, \cdot] = \left( \hatg_{\ell, a} + \Oti(m n^2 \mache) \right) P_1 P_2 \ldots P_j \;,
\end{equation}
where the $P_i$ are ideal Householder reflectors for subvectors of the computed~$\MF$, as specified in the algorithm.
As before, define $F_j := P_j P_{j-1} \cdots P_1$.
Then it follows directly from~\eqref{eq-error-ghat1}:
\begin{equation} \label{eq-error-ghat2}
 \MF(a)[\ell, \cdot] F_j = \hatg_{\ell, a} + \Oti(m n^2 \mache)
\end{equation}
By combining \eqref{eq-error-def-g}, \eqref{eq-error-g} and~\eqref{eq-error-ghat2} we obtain
\begin{equation} \label{eq-error-flip}
 \hatf_\ell M(a) = \MF(a)[\ell, \cdot] F_j + \Oti(m n^2 \mache)
\end{equation}
Using again the fact that the first $j$ rows of~$F$ coincide with the first $j$ rows of~$F_j$, we have as in~\eqref{eq-slices}:
\begin{equation} \label{eq-error-slices}
\begin{aligned}
\MF(a)[\ell, \cdot] F_j & = \MF(a)[\ell, 1..\nF] F - \MF(a)[\ell, j{+}1..\nF] F[j{+}1..\nF,\cdot] \\
 & \hspace{29mm}                                   + \MF(a)[\ell, j{+}1..n] F_j[j{+}1..n,\cdot]
\end{aligned}
\end{equation}
Using~\eqref{eq-error-F} we have:
\begin{equation} \label{eq-error-F-applied}
 \MF(a)[\ell, 1..\nF] F = \MF(a)[\ell, 1..\nF] \hatF + \Oti(m n^{2.5} \mache)
\end{equation}
By combining \eqref{eq-error-flip}, \eqref{eq-error-slices} and~\eqref{eq-error-F-applied} we get:
\begin{equation}
\begin{aligned} \label{eq-error-bigcombined}
 & \quad \hatf_\ell M(a) - \MF(a)[\ell, 1..\nF] \hatF \\
 & = - \MF(a)[\ell, j{+}1..\nF] F[j{+}1..\nF,\cdot] \\
 & \quad \mbox{} + \MF(a)[\ell, j{+}1..n] F_j[j{+}1..n,\cdot] + \Oti(m n^{2.5} \mache)
\end{aligned}
\end{equation}
We have:
\begin{align*}
 & \norm{\E(a)[\ell, \cdot]} \\
 & = \norm{\hatf_\ell M(a) - \MF(a)[\ell, 1..\nF] \hatF} && \text{def.\ of~$\E(a)$} \\
 & \le \norm{\MF(a)[\ell, j{+}1..\nF] F[j{+}1..\nF,\cdot] + \MF(a)[\ell, j{+}1..n] F_j[j{+}1..n,\cdot]} && \text{by~\eqref{eq-error-bigcombined}} \\
 & \quad \mbox{} + c m n^{2.5} \mache \\
 & \le 2 \tau + c m n^{2.5} \mache && \text{as in~\eqref{eq-Delta-bound}}
\end{align*}
From this row-wise bound we get the desired bound on the matrix norm, see \cite[Lemma~6.6.a]{Higham}:
\[
 \norm{\E(a)}_2 \le 2 \sqrt{n} \tau + c m n^{3} \mache
\]
\qed
\end{proof}

\subsection{Proof of Theorem~\ref{thm-error-bound}} \label{app-error}

\begin{qtheorem}{\ref{thm-error-bound}}
\stmtthmerrorbound
\end{qtheorem}
Part~1.\ follows from Proposition~\ref{prop-Delta-bound}.1.
Part~2.\ follows from Proposition~\ref{prop-Delta-bound}.2.\ and Theorem~\ref{thm-WA-minimal}.
It remains to prove part~3.
We use again the notation $\hatF$ for the computed version of~$F$, as in the proof of Proposition~\ref{prop-Delta-bound}.
We have the following lemma:
\begin{lemma} \label{lem-error-induction}
 Consider HouseholderReduction, see Figure~\ref{fig-householder-tzeng}.
 Let $m > 0$ such that $\norm{M(a)} \le m$ and $\norm{\MF(a)} \le m$ hold for all $a \in \Sigma$.
 Let $b := 2 \sqrt{n} \tau + c m n^{3} \mache$ be the bound from Proposition~\ref{prop-Delta-bound}.
 For all $w \in \Sigma^*$ we have:
 \[
  \norm{\hatF M(w) - \MF(w) \hatF} \le b |w| m^{|w|-1}
 \]
\end{lemma}
\begin{proof}
We proceed by induction on~$|w|$.
The base case, $|w| = 0$, is trivial.
Let $|w| \ge 0$ and $a \in \Sigma$.
With the matrix $\E(a)$ from Proposition~\ref{prop-Delta-bound} we have:
\begin{align*}
 \hatF M(a w) - \MF(a w) \hatF
 & = \hatF M(a) M(w) - \MF(a) \MF(w) \hatF \\
 & = \big(\MF(a) \hatF + \E(a)\big) M(w) - \MF(a) \MF(w) \hatF \\
 & = \MF(a) \big(\hatF M(w) - \MF(w) \hatF \big) + \E(a) M(w)
\end{align*}
Using the induction hypothesis, Proposition~\ref{prop-Delta-bound}, and the bounds on $\norm{\MF(a)} \le m$ and $\norm{M(w)} \le m^{|w|}$ we obtain:
\[
 \norm{\hatF M(a w) - \MF(a w) \hatF}
 \ \le \ m b |w| m^{|w|-1} + b m^{|w|} \ = \ b (|w|+1) m^{|w|}
\]
\qed
\end{proof}

Now we can prove part~3.\ of Theorem~\ref{thm-error-bound}:
\begin{proof}[of Theorem~\ref{thm-error-bound}, part~3.]
We use again the $\Oti$-notation from Proposition~\ref{prop-Delta-bound}.
The vector $\hatf_1$ (the first row of~$\hatF$) is computed by applying one Householder reflector to~$e(1)$.
So we have by~\cite[Lemma~19.3]{Higham}:
\[
 \hatf_1 = \frac{\pm \alpha}{\norm{\alpha}} + \Oti(n \mache)
\]
Hence it follows:
\begin{equation} \label{eq-error-alpha}
 \alphaF \hatF = \alphaF[1] \hatf_1 = \alpha + \Oti( \norm{\alpha} n \mache)
\end{equation}

The vector $\etaF$ is computed by multiplying $\hatF$ with~$\eta$.
So we have by~\cite[Chapter~3.5]{Higham}:
\begin{equation} \label{eq-error-eta}
 \etaF = \hatF \eta + \Oti(\norm{\eta} n^{1.5} \mache)
\end{equation}

Let $w \in \Sigma^*$.
We have:
\begin{align*}
 & |L_\A(w) - L_{\AF}(w)| \\
 & = |\alpha M(w) \eta - \alphaF \MF(w) \etaF| \\
 & = |\alphaF \hatF M(w) \eta - \alphaF \MF(w) \hatF \eta| + \Oti( \norm{\alpha} m^w \norm{\eta} n^{1.5} \mache )
  && \text{by \eqref{eq-error-alpha}, \eqref{eq-error-eta}} \\
 & = |w| \norm{\alpha} \left(2 \sqrt{n} \tau + c m n^{3} \mache\right) m^{|w|-1} \norm{\eta} + \Oti( \norm{\alpha} m^w \norm{\eta} n^{1.5} \mache )
  && \text{Lemma~\ref{lem-error-induction}} \\
 & = 2 |w| \norm{\alpha} m^{|w|-1} \norm{\eta} \sqrt{n} \tau + \Oti( \max\{|w|,1\} \norm{\alpha} m^{|w|} \norm{\eta} n^3 \mache )
\end{align*}
One can show the same bound on $|L_{\AF}(w) - L_{\A'}(w)|$ in the same way.
The statement follows.
\qed
\end{proof}

\subsection{Image Compression} \label{app-image-compression}

We have implemented the algorithm of Theorem~\ref{thm-error-bound} in a prototype tool
in C++ using the Boost uBLAS library, which provides basic matrix and vector data structures.
Further, we have built a translator between images and automata:
it loads compressed
images, constructs an automaton based on recursive algorithm sketched in the main body of the paper,
feeds this automaton to the minimiser, reads back the minimised automaton,
and displays the resulting compressed image.

\begin{figure}
\mbox{} \hspace{-20mm}
\begin{tabular}{p{0.6\textwidth}p{0.5\textwidth}}
\subfloat[original image: \ 33110 states \label{beehive:original}]{\includegraphics[scale=0.4]{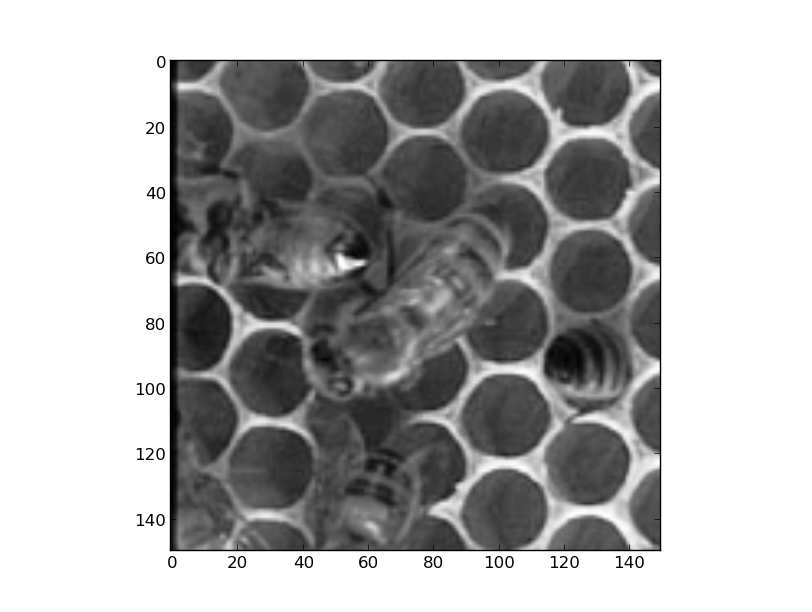}}  &
\subfloat[$\tau=10^{-6}$: \ 229 states \label{beehive:mimimal_pretty}]{\includegraphics[scale=0.4]{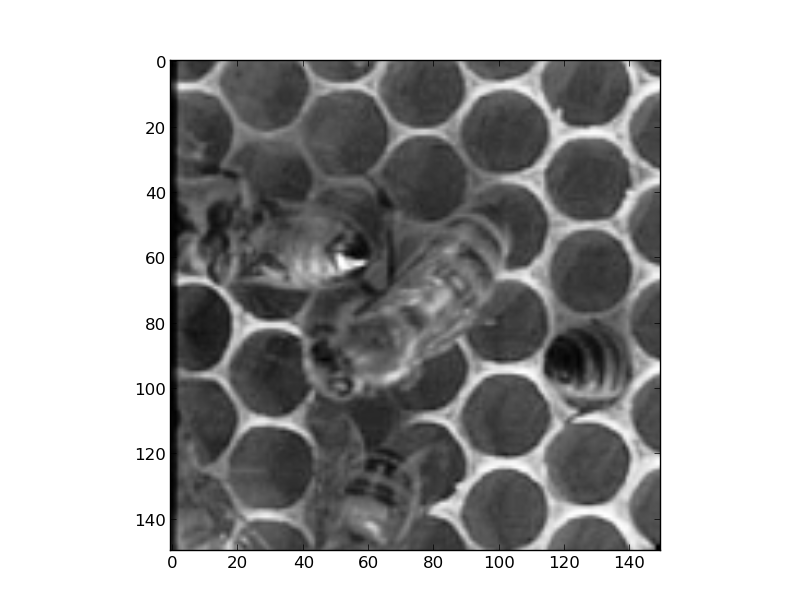}} \\
\subfloat[$\tau=1.5 \times 10^{-2}$: \ 175 states \label{beehive:mimimal_fuzzy}]{\includegraphics[scale=0.4]{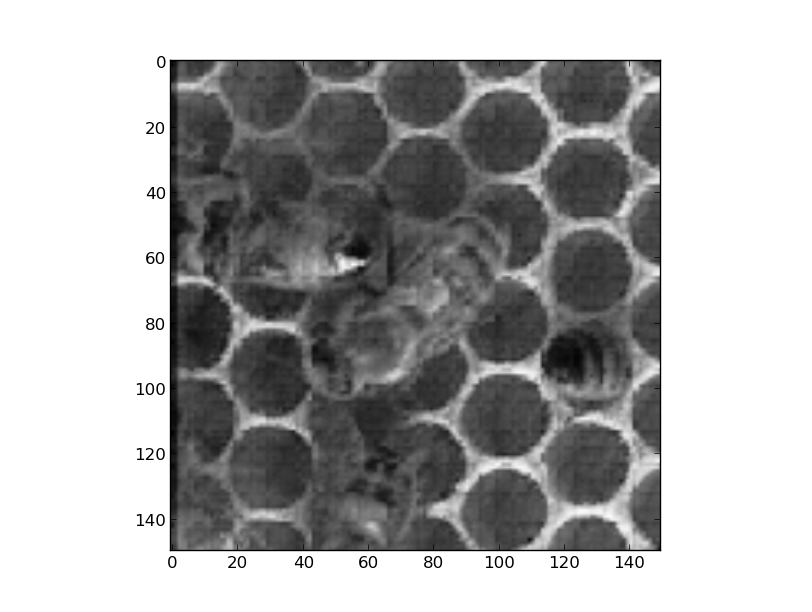}} &
\subfloat[$\tau=2 \times 10^{-2}$: \ 121 states \label{beehive:mimimal_blurred}]{\includegraphics[scale=0.4]{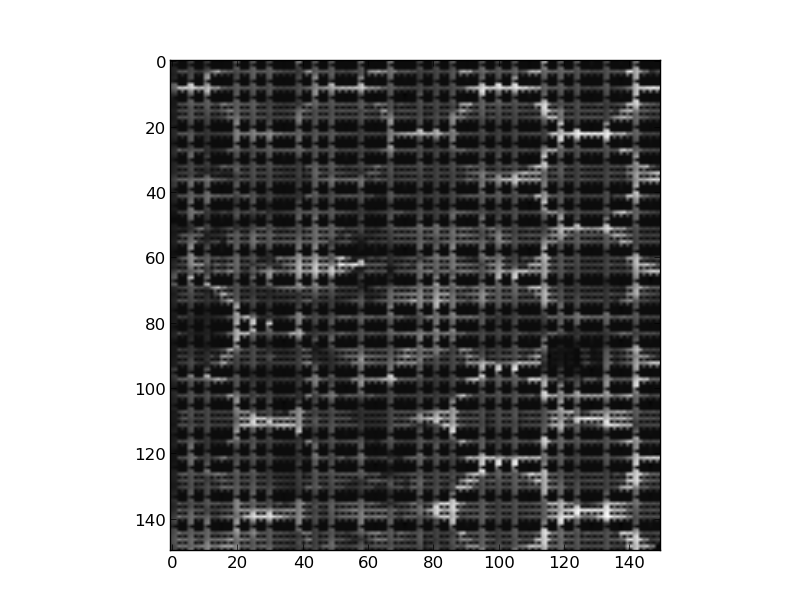}} \\
\end{tabular}
\caption{Image compression via WA minimisation with different values of the error tolerance parameter~$\tau$.\label{beehive}}
\end{figure}

We have applied our image compression tool to some images.
Figure~\ref{beehive:original} shows a picture with a resolution of $150 \times 150$ pixels.
The weighted automaton that encodes the image exactly has 33110 states.
Applying minimisation with error tolerance parameter $\tau=10^{-6}$ yields an automaton
with 229 states, which leads to the compressed picture in Figure~\ref{beehive:mimimal_pretty}.
Larger values of~$\tau$ lead to smaller automata and blurrier pictures,
 see Figures \ref{beehive:mimimal_fuzzy} and~\ref{beehive:mimimal_blurred},
 where the pictures change perceivably.

We remark that our tool is not meant to deliver state-of-the-art image compression, which would require many tweaks, as indicated in~\cite{CulikKari93}.

\section{Proofs of Section~\ref{sec-minpro}} \label{app-minpro}

\subsection{Continuation of the proof of Proposition~\ref{prop-redCubeAut} } \label{app-prop-redCubeAut}

\begin{qproposition}{\ref{prop-redCubeAut}}
 \stmtpropredCubeAut
\end{qproposition}
\begin{proof}
Consider the reduction given in the main body of the paper.
Observe that we have for all $i \in \{2, \ldots, k\}$ and all $s \in \set{d}$ that
\begin{equation}
  L_\A(a_i b_s) \ = \  M(a_i)[1,i] \ M(b_s)[i,k+1] \ = \ p_i[s]\,. \label{eq-redCubeAut0}
\end{equation}
It remains to show the correctness of the reduction, i.e., we need to show that
 there is a set $Q = \{q_1, \ldots, q_\ell\} \subseteq [0,1]^d$ with $\conv(Q) \supseteq P$
  if and only if
 there is a PA $\A' = (\ell+1, \Sigma, M', \alpha', \eta')$ equivalent to~$\A$.

For the ``only if'' direction, suppose that $Q = \{q_1, \ldots, q_\ell\}$ such that $\conv(Q) \supseteq P$.
As $p_1 = (0, \ldots, 0) \in P$ is a vertex of the hypercube, we have $p_1 \in Q$, say $q_1 = (0,\ldots,0)$.
As $\conv(Q) \supseteq P$, for each $i \in \{2, \ldots, k\}$
 there are $\lambdas{i}_1, \ldots, \lambdas{i}_\ell \ge 0$ with $\sum_{j=1}^\ell \lambdas{i}_j = 1$ and
 \begin{equation} \label{eq-redCubeAut1}
  p_i = \sum_{j=1}^\ell \lambdas{i}_j q_j = \sum_{j=2}^\ell \lambdas{i}_j q_j\,.
 \end{equation}
Build the PA $\A' = (\ell+1, \Sigma, M', \alpha', \eta')$ as follows.
Set $M'(a_i)[1, j] := \lambdas{i}_j$ and
   $M'(b_s)[j, \ell+1] := q_j[s]$ for all $i \in \{2, \ldots, k\}$ and all $j \in \{2, \ldots, \ell\}$ and all $s \in \set{d}$,
 and set all other entries of~$M'$ to~$0$.
Set $\alpha' := e(1)$ and $\eta' := e(\ell+1)^T$.
Then $\A, \A'$ are equivalent, as
 \begin{align*}
  L_\A(a_i b_s)
   &\stackrel{\eqref{eq-redCubeAut0}}{=} p_i[s]
   \stackrel{\eqref{eq-redCubeAut1}}{=} \sum_{j=2}^\ell \lambdas{i}_j q_j[s]
   = \sum_{j=2}^\ell M'(a_i)[1,j] \ M'(b_s)[j,\ell+1]
   = L_{\A'}(a_i b_s) \,.
 \end{align*}

For the ``if direction'', suppose that $\A' = (\ell+1, \Sigma, M', \alpha', \eta')$ is a PA with $L_{\A'} = L_{\A}$.
For any vector $\beta \in [0,1]^{\ell+1}$ we define $\supp(\beta) := \{j \in \set{\ell+1} \mid \beta[j] > 0\}$.
Define the following subsets of~$\set{\ell+1}$:
\begin{align*}
 J_1 & := \supp(\alpha') \cap \bigcup_{i \in \{2, \ldots, k\}, \ s \in \set{d}} \supp\left( M'(a_i) M'(b_s) \eta' \right) \\
 J_2 & := \bigcup_{i \in \{2, \ldots, k\}} \supp\left(\alpha' M'(a_i) \right) \cap \bigcup_{s \in \set{d}} \supp\left( M'(b_s) \eta' \right) \\
 J_3 & := \bigcup_{i \in \{2, \ldots, k\}, \ s \in \set{d}} \supp\left(\alpha' M'(a_i) M'(b_s) \right) \cap \supp( \eta' )
\end{align*}
Recall that $L_{\A'} = L_{\A}$.
Since $L_{\A}(b_s) = 0$ for all~$s$, we have $J_1 \cap J_2 = \emptyset$.
Since $L_{\A}(\tau) = 0$, we have $J_1 \cap J_3 = \emptyset$.
Since $L_{\A}(a_i) = 0$ for all~$i$, we have $J_2 \cap J_3 = \emptyset$.
If one of $J_1, J_2, J_3$ is the empty set, then $J_1 = J_2 = J_3 = \emptyset$ and we have $L_{\A}(w) = 0$ for all $w \in \Sigma^*$,
 so then by~\eqref{eq-redCubeAut0} we have $P = \{(0, \ldots, 0)\}$, and one can take $Q = P$.
So we can assume for the rest of the proof that $J_1, J_2, J_3$ are all non-empty.
But they are pairwise disjoint, so it follows $|J_2| \le \ell-1$.
Without loss of generality, assume $1 \not\in J_2$.
For $i \in \{2, \ldots, k\}$ and $j \in J_2$, define $\lambdas{i}_j \ge 0$ and $q_j \in [0,1]^d$ with
\[
 \lambdas{i}_j = \left( \alpha' M'(a_i) \right)[j] \quad \text{and} \quad q_j[s] = \left( M'(b_s) \eta' \right)[j] \quad \text{for $s \in \set{d}$.}
\]
Let $q_1 := (0, \ldots, 0)$.
Since $\alpha' M'(a_i)$ is stochastic, one can choose $\lambdas{i}_1 \ge 0$ so that $\sum_{j \in \{1\} \cup J_2} \lambdas{i}_j = 1$.
We have:
\begin{align*}
 p_i[s]
 & \stackrel{\eqref{eq-redCubeAut0}}{=} L_{\A}(a_i b_s) = L_{\A'}(a_i b_s) = \alpha' M'(a_i) M'(b_s) \eta' \\
 & = \sum_{j \in J_2} \left( \alpha' M'(a_i) \right)[j] \ \left( M'(b_s) \eta' \right)[j] \quad = \ \sum_{j \in \{1\} \cup J_2} \lambdas{i}_j q_j[s]
\end{align*}
It follows that $P \subseteq \conv(Q)$ holds for $Q := \{q_j \mid j \in \{1\} \cup J_2 \}$, with $|Q| \le \ell$.
\qed
\end{proof}

\subsection{Proof of Proposition~\ref{prop-cube-NPhardness}} \label{app-prop-cube-NPhardness}

\begin{qproposition}{\ref{prop-cube-NPhardness}}
\stmtpropcubeNPhardness
\end{qproposition}

\begin{proof}
We reduce 3SAT to the hypercube problem.
Let $x_1, \ldots, x_N$ be the variables and let $\varphi = c_1 \land \ldots \land c_M$ be a 3SAT formula.
Each clause~$c_j$ is a disjunction of three literals $c_j = l_{j,1} \lor l_{j,2} \lor l_{j,3}$,
 where $l_{j,k} = x_{j,k}^-$ or $l_{j,k} = x_{j,k}^+$ and $x_{j,k} \in \{x_1, \ldots, x_N\}$.
(It is convenient in the following to distinguish between a variable and a positive literal,
 so we prefer the notation $x_i^-$ and $x_i^+$ over the more familiar $\neg x_i$ and $x_i$ for literals.)
We can assume that no clause appears twice in~$\varphi$ and that no variable appears twice in the same clause.
Define a set $D$ of coordinates:
 \[
  D := \{\spco{x_i}, y_i, z_i \mid i \in \set{N}\} \cup \{\spco{c_j} \mid j \in \set{M}\}
 \]
We take $d := |D| = 3 N + M$.
For $u \in D$ denote by $e(u) \in \{0,1\}^D$ the vector with $e(u)[u] = 1$ and $e(u)[u'] = 0$ for $u' \in D \setminus \{u\}$.
For $i \in \set{N}$, define shorthands $f(x_i^-) := e(y_i)$ and $f(x_i^+) := e(y_i) + e(z_i)$.
Observe that those points are vertices of the hypercube.

Define:
\begin{align*}
 P_{\mathit{var}}
 & := \{e(\spco{x_i}) + f(x_i^-), \ e(\spco{x_i}) + f(x_i^+) \mid i \in \set{N}\} \\
 P_{\mathit{cla}}
 & := \{e(\spco{c_j}) + f(l_{j,1}), \ e(\spco{c_j}) + f(l_{j,2}), \ e(\spco{c_j}) + f(l_{j,3}) \mid j \in \set{M}\} \\
 p(x_i)
 & := \frac12 e(\spco{x_i}) + e(y_i) + \frac12 e(z_i) \\
 & \ = \frac12 e(\spco{x_i}) + \frac12 f(x_i^-) + \frac12 f(x_i^+) \text{ \hspace{22.5mm} for $i \in \set{N}$} \\
 p(c_j)
 & := \frac23 e(\spco{c_j}) + \frac13f(l_{j,1}) + \frac13 f(l_{j,2}) + \frac13 f(l_{j,3}) \text{ \qquad for $j \in \set{M}$} \\
 P & := P_{\mathit{var}} \cup P_{\mathit{cla}} \cup \{p(x_1), \ldots, p(x_N), \ p(c_1), \ldots, p(c_M)\}
\end{align*}
Figure~\ref{fig-cube-NPhardness} visualizes the points in~$P$.

\begin{figure}
\begin{tikzpicture}[scale=2.5]
\coordinate (fxm) at (0,0);
\coordinate (fxp) at (1,0);
\coordinate (gxm) at (0,1);
\coordinate (gxp) at (1,1);
\coordinate (p) at (0.5,0.5);
\draw[fill] (gxm) circle (0.03);
\draw[fill] (gxp) circle (0.03);
\draw[fill] (p) circle (0.03);
\draw (fxm) -- (fxp) -- (gxp) -- (gxm) -- (fxm);
\node at ($(fxm) + (0,-0.1)$) {$f(x_i^-)$};
\node at ($(fxp) + (0,-0.1)$) {$f(x_i^+)$};
\node at ($(p) + (0,-0.15)$) {$p(x_i)$};
\node at ($(gxm) + (-0.0,+0.15)$) {$e(\spco{x_i}) + f(x_i^-)$};
\node at ($(gxp) + (+0.0,+0.15)$) {$e(\spco{x_i}) + f(x_i^+)$};
\end{tikzpicture}
\hfill
\begin{tikzpicture}[scale=2.5]
\coordinate (p1) at (0.8,0);
\coordinate (p2) at (0.1,0.2);
\coordinate (p3) at (-0.7,-0.3);
\coordinate (p23) at ($1/2*(p2) + 1/2*(p3)$);
\coordinate (p13) at ($1/2*(p1) + 1/2*(p3)$);
\coordinate (p12) at ($1/2*(p1) + 1/2*(p2)$);
\node at ($(p1)+(0.05,-0.15)$) {$e(\spco{c_j}) + f(l_{j,1})$};
\node at ($(p2)+(0.,0.1)$) {$e(\spco{c_j}) + f(l_{j,2})$};
\node[left] at ($(p3)+(0.0,0.0)$) {$e(\spco{c_j}) + f(l_{j,3})$};
\draw (p1) -- (p2) -- (p3) -- (p1);
\foreach \i in {1,2,3} {
 \draw[fill] (p\i) circle (0.03);
}
\foreach \i in {1,2,3} {
 \coordinate (dp\i) at ($(p\i) + (0,-1)$);
}
\draw (dp1) -- (dp2) -- (dp3) -- (dp1);
\draw[dashed] (p23) -- (dp1);
\draw[dashed] (p13) -- (dp2);
\draw[dashed] (p12) -- (dp3);
\coordinate (center) at ($1/3*(dp1) + 2/3*(p23)$);
\draw[fill] (center) circle (0.03);
\node[right] at ($(center) + (0.05,0.02)$) {$p(c_j)$};
\node at ($(dp1)+(0.2,-0.0)$) {$f(l_{j,1})$};
\node at ($(dp2)+(0,-0.15)$) {$f(l_{j,2})$};
\node[left] at ($(dp3)+(0.0,0.05)$) {$f(l_{j,3})$};
\end{tikzpicture}
\caption{Reduction from 3SAT to the hypercube problem.
The left figure visualizes the $\{z_i, \spco{x_i}\}$-face of the hypercube
 with the $y_i$-coordinate $=1$ and all other coordinates $=0$.
The black points are in~$P$.
Observe that $p(x_i) \in \conv(\{e(\spco{x_i}) + f(x_i^-), \ f(x_i^+)\})$ and
             $p(x_i) \in \conv(\{e(\spco{x_i}) + f(x_i^+), \ f(x_i^-)\})$.
\newline
The right figure visualizes six hypercube vertices and a point $p(c_j) \in P$.
The black points are in~$P$.
Observe that $p(c_j) \in \conv(\{e(\spco{c_j}) + f(l_{j,k(1)}), \ e(\spco{c_j}) + f(l_{j,k(2)}), \ f(l_{j,k(3)})\})$
 for all $k(1), k(2), k(3)$ with $\{k(1), k(2), k(3)\} = \{1,2,3\}$.
}
\label{fig-cube-NPhardness}
\end{figure}
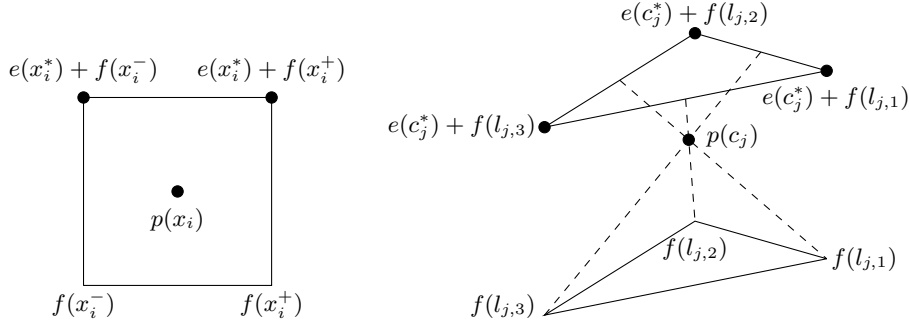

Observe that $|P| = 3 N + 4 M$.
Take $\ell := 3 N + 3 M$.

First we show that if $\varphi$ is satisfiable, then there is a set~$Q \subseteq [0,1]^d$ with $|Q| \le \ell$ and $\conv(Q) \supseteq P$.
Let $\sigma: \{x_1, \ldots, x_N\} \to \{\mathit{true}, \mathit{false}\}$ be an assignment that satisfies~$\varphi$.
Define:
\begin{align*}
 s_i & := \begin{cases} f(x_i^-) & \text{ if $\sigma(x_i) = \mathit{false}$} \\
                        f(x_i^+) & \text{ if $\sigma(x_i) = \mathit{true}$}
          \end{cases} \qquad \text{ for $i \in \set{N}$}\\
 Q & := P_{\mathit{var}} \cup P_{\mathit{cla}} \cup \{ s_1, \ldots, s_N\}
\end{align*}
We have $|Q| = 3 N + 3 M$.
Clearly, $\conv(Q) \supseteq P_{\mathit{var}} \cup P_{\mathit{cla}}$.
Moreover:
\begin{itemize}
 \item Let $i \in \set{N}$.
       If $\sigma(x_i) = \mathit{false}$, then $p(x_i) = \frac12 \big(e(\spco{x_i}) + f(x_i^+)\big) + \frac12 s_i$;
       if $\sigma(x_i) = \mathit{true}$, then $p(x_i) = \frac12 \big(e(\spco{x_i}) + f(x_i^-)\big) + \frac12 s_i$.
 \item Let $j \in \set{M}$.
       As $\sigma$ satisfies~$\varphi$, there are $k(1),k(2),k(3)$ such that $\{k(1), k(2), k(3)\} \in \{1,2,3\}$ and $\sigma(l_{j,k(1)}) = \mathit{true}$.
       Let $i \in \set{N}$ such that $l_{j,k(1)} \in \{x_i^-, x_i^+\}$.
       Then $p(c_j) = \frac13 \big(e(\spco{c_j}) + f(l_{j,k(2)})\big) + \frac13 \big(e(\spco{c_j}) + f(l_{j,k(3)})\big) + \frac13 s_i$.
\end{itemize}
Hence we have that $\conv(Q) \supseteq P$.

For the converse, let $Q \subseteq [0,1]^d$ with $|Q| \le 3 N + 3 M$ and $\conv(Q) \supseteq P$.
Then $Q \supseteq P_{\mathit var} \cup P_{\mathit cla}$, as $P_{\mathit var}$ and $P_{\mathit cla}$ consist of hypercube vertices.

Let $i \in \set{N}$.
Let $Q_i^{\mathit var} \subseteq Q$ be a minimal subset of~$Q$ with $p(x_i) \in \conv(Q_i^{\mathit var})$,
 i.e., if $Q_i'$ is a proper subset of~$Q_i^{\mathit var}$, then $p(x_i) \not\in \conv(Q_i')$.
As $p(x_i)[\spco{x_{i'}}] = p(x_i)[\spco{c_j}] = 0$ holds for all $i' \in \set{N} \setminus \{i\}$ and all $j \in \set{M}$, we have
 \[
  Q_i^{\mathit var} \cap (P_{\mathit var} \cup P_{\mathit cla}) \quad \subseteq \quad \{e(\spco{x_i}) + f(x_i^-), \ e(\spco{x_i}) + f(x_i^+)\} \,.
 \]
As $p(x_i)[\spco{x_i}] = \frac12$ and $p(x_i)[y_i] = 1$, there is a point $s_i \in Q_i^{\mathit var}$ with
 $s_i[\spco{x_i}] \le \frac12$ and $s_i[y_i] = 1$ and $s_i[z_i] \in [0,1]$ and $s_i[u] = 0$ for all other coordinates~$u \in D$.

It follows that $Q = P_{\mathit var} \cup P_{\mathit cla} \cup \{s_1, \ldots, s_N\}$ and $|Q| = 3 N + 3 M$.
Let $\sigma$ be any assignment with
 \[
  \sigma(x_i) = \begin{cases} \mathit{false} & \text{ if $s_i[z_i] = 0$} \\
                              \mathit{true}  & \text{ if $s_i[z_i] = 1$ \,.}
                \end{cases}
 \]
We show that $\sigma$ satisfies~$\varphi$.
Let $j \in \set{M}$.
Let $Q_j^{\mathit cla} \subseteq Q$ be a minimal subset of~$Q$ with $p(c_j) \in \conv(Q_j^{\mathit cla})$,
 i.e., if $Q_j'$ is a proper subset of~$Q_j^{\mathit cla}$, then $p(c_j) \not\in \conv(Q_j')$.
As $p(c_j)[\spco{c_{j'}}] = p(c_j)[\spco{x_i}] = 0$ holds for all $j' \in \set{M} \setminus \{j\}$ and all $i \in \set{N}$,
 we have
 \[
  Q_j^{\mathit cla} \subseteq \{ e(\spco{c_j}) + f(l_{j,1}), \ e(\spco{c_j}) + f(l_{j,2}), \ e(\spco{c_j}) + f(l_{j,3})\} \cup \{s_1, \ldots, s_N\}\,.
 \]
As $p(c_j)[\spco{c_j}] = \frac23 < 1$, there exists an~$i$ such that $s_i \in Q_j^{\mathit cla}$.
As $s_i[y_i] = 1 > 0$, we have $p(c_j)[y_i] > 0$.
Hence the variable~$x_i$ appears in~$c_j$, so one of the following two cases holds:
\begin{itemize}
\item
 The literal $x_i^-$ appears in~$c_j$.
 As we have $p(c_j)[z_i] = 0$, it follows that $q[z_i] = 0$ holds for all $q \in Q_j^{\mathit cla}$.
 In particular, we have $s_i[z_i] = 0$, so $\sigma(x_i) = \mathit{false}$.
\item
 The literal $x_i^+$ appears in~$c_j$.
 Note that for all points $q \in Q_j^{\mathit cla}$ we have $q[y_i] \ge q[z_i]$.
 As we have $p(c_j)[y_i] = \frac13 = p(c_j)[z_i]$, it follows that $q[y_i] = q[z_i]$ holds for all $q \in Q_j^{\mathit cla}$.
 In particular, we have $s_i[z_i] = 1$, so $\sigma(x_i) = \mathit{true}$.
\end{itemize}
For both cases it follows that $\sigma$ satisfies~$c_j$.
As $j$ was chosen arbitrarily, we conclude that $\sigma$ satisfies~$\varphi$.
This completes the reduction to the hypercube problem.

The given reduction does not put the origin in~$P$.
However, $P_{\mathit{var}} \cup P_{\mathit{cla}} \subseteq P$ consist of corners of the hypercube.
One can pick one of the corners in~$P$ and apply a simple linear coordinate transformation to all points in~$P$ such that
 the picked corner becomes the origin.
Hence the restricted hypercube problem is NP-hard as well.
\qed
\end{proof}

\subsection{Proof of Proposition~\ref{prop-PSPACE}} \label{app-prop-PSPACE}
\begin{qproposition}{\ref{prop-PSPACE}}
\stmtpropPSPACE
\end{qproposition}
\begin{proof}
We say, a WA $\A$ is \emph{zero} if $L_\A(w) = 0$ holds for all $w \in \Sigma^*$.
For two WAs $\A_i = (n_i, \Sigma, M_i, \alpha_i, \eta_i)$ (with $i=1,2$), define their \emph{difference WA}
 $\A = (n, \Sigma, M, \alpha, \eta)$, where $n = n_1 + n_2$, and $M(a) = \begin{pmatrix} M_1(a) & 0 \\ 0 & M_2(a) \end{pmatrix}$ for $a \in \Sigma$,
  and $\alpha = (\alpha_1, \alpha_2)$, and $\eta = (\eta_1^T, - \eta_2^T)^T$.
Clearly, $L_A(w) = L_{\A_1}(w) - L_{\A_2}(w)$ holds for all $w \in \Sigma^*$.
So WAs $\A_1$~and~$\A_2$ are equivalent if and only if their difference WA is zero.

A WA $\A = (n, \Sigma, M, \alpha, \eta)$ is zero if and only if
 all vectors of the forward space $\langle \alpha M(w) \mid w \in \Sigma^* \rangle$ are orthogonal to~$\eta$
 (see, e.g., \cite{Tzeng}).
It follows that a WA $\A = (n, \Sigma, M, \alpha, \eta)$ is zero if and only if there is a vector space $\FV \subseteq \R^n$
 with $\alpha \in \FV$, and $\FV \eta = \{0\}$, and $\FV M(a) \subseteq \FV$ holds for all $a \in \Sigma$.
(Here, the actual forward space is a subset of~$\FV$.)

The proposition follows from those observations.
\qed
\end{proof}

}{}
\end{document}